\documentclass{article}

\usepackage{arxiv}

\usepackage[utf8]{inputenc} 
\usepackage[T1]{fontenc}    
\usepackage{newunicodechar} 
\newunicodechar{ğ}{\u{g}}
\usepackage{url}            
\usepackage{booktabs}       
\usepackage{amsfonts}       
\usepackage{amsthm}         
\usepackage{nicefrac}       
\usepackage{microtype}      
\usepackage{makecell}      
\usepackage{multirow}
\usepackage{lipsum}
\usepackage{graphicx}
\usepackage[style=authoryear,maxnames=10,maxcitenames=2,giveninits=true,backend=biber]{biblatex}
\usepackage{amsmath}
\usepackage{caption}
\usepackage{subcaption}
\usepackage[ruled,linesnumbered]{algorithm2e}
\usepackage{setspace}
\usepackage{tikz}
\usepackage{tcolorbox}
\usepackage{amsthm}
\usepackage{amsmath}
\usepackage{amssymb}
\usepackage{booktabs}
\usepackage{multirow}
\usepackage{float}
\usepackage{subcaption}
\usepackage{array}
\usepackage{rotating}
\usepackage{booktabs}
\usepackage{etoolbox}
\usepackage{pdflscape}
\usepackage{graphicx}
\usepackage{caption}
\usepackage{makecell}
\usepackage{rotating}
\usepackage{mathrsfs}
\usepackage[export]{adjustbox}
\usepackage{todonotes}
\usepackage{hyperref}       
\usepackage{standalone}

\usepackage{tikz}
\usetikzlibrary{calc, arrows.meta}

\newcommand{\A}{\Lambda}

\newcommand{\algname}[1]{\textsc{#1}}
\newcommand{\const}[1]{\mathsf{#1}}

\providecommand{\floor}[1]{\ensuremath{\left\lfloor #1\right\rfloor}}

\def\theoremcounterhook{section}

\newtheorem{theorem}{Theorem}[section]
\newtheorem{definition}{Definition}
\newtheorem{proposition}{Proposition}

\title{Wildfire Suppression: Complexity, Models, and Instances}

\author{
 Gustavo Delazeri \\
  Institute of Informatics\\
  Universidade Federal do Rio Grande do Sul\\
  Porto Alegre, Brazil \\
  \texttt{gustavo.delazeri@inf.ufrgs.br} \\
   \And
 Marcus Ritt \\
  Institute of Informatics\\
  Universidade Federal do Rio Grande do Sul\\
  Porto Alegre, Brazil \\
  \texttt{marcus.ritt@inf.ufrgs.br} \\
}

\begin{document}
\maketitle

\begin{abstract}
  Wildfires cause major losses worldwide, and the frequency of fire-weather conditions is likely to increase in many regions. We study the allocation of suppression resources over time on a graph-based representation of a landscape to slow down fire propagation. Our contributions are theoretical and methodological. First, we prove strong NP-completeness on planar graphs for this problem and two related variants, and on full weighted directed grids for two of the three problems. We also show that this problem remains strongly NP-complete when all resources are released simultaneously. Second, we propose a new mixed-integer programming (MIP) formulation that obtains state-of-the-art results, showing that MIP is a competitive approach contrary to earlier findings. Third, showing that existing benchmarks lack realism and difficulty, we introduce a physics-grounded instance generator based on Rothermel's surface fire spread model. We use these diverse instances to benchmark the literature, identifying the specific conditions where each algorithm succeeds or fails.
\end{abstract}

\section{Introduction}

Wildfires cause major damage globally, as reflected in insured losses: in January 2025, fires in Los Angeles County alone resulted in an estimated USD 40 billion in damages, while Australia's 2019--20 bushfires totaled AUD 1.866 billion~\parencite{SwissRe/2025, PERILS/2021}. In addition to these damages, the cost of wildfire suppression is also increasing. In the United States, firefighting costs rose from about USD 240 million in 1985 to USD 3.17 billion in 2023~\parencite{NIFCSuppressionCosts}. The frequency and intensity of extreme wildfire events have doubled over the last two decades, and the occurrence of fire-prone weather has become more likely in several regions worldwide, with further increases projected in other regions~\parencite{IPCC_2022_WGII,Cunningham/2024}.

Research into wildfire prevention and combat dates back to the 1960s~\parencite{Martell/1982}. \textcite{Jewell/1963} was among the first to model the economic trade-offs of wildfire management. He represented the fire as a linear front spreading with constant velocity along a single axis, which had to be suppressed by a firebreak perpendicular to the spread. The goal is to minimize the total cost of wages and equipment for firefighters and the value of land lost. The resulting model is simple enough to be solved analytically.

Diverse mathematical models address decisions in wildfire prevention and suppression. \textcite{Mendes/2022} categorize these based on the type of decision they make. For example, vehicle routing models determine the optimal dispatch and timing of resources to assets (e.g., water tankers to factories) to minimize fire damage~\parencite{Merwe/2015}. Covering models address decisions about strategically locating resources, such as vehicles or stations, to maximize the potential for rapid initial attack across fire-prone areas~\parencite{Dimopoulou/2001}.

Graph-based (or grid-based) models represent wildfire spread by discretizing the landscape into a grid of basic units or cells. This is encoded as a graph, where each cell is a vertex and arcs connect vertices of adjacent cells. The arcs are labeled with the estimated fire travel time from the center of a cell to the center of an adjacent cell, which depends on local conditions such as fuel type, wind speed, and slope. Fire propagation across the landscape can then be approximated by calculating the shortest paths from an ignition point to other cells. If a fire suppression resource is deployed to a cell, fire spread through that cell is delayed, similar to the effect of a firebreak. Integrating fire spread and suppression actions into a single optimization model allows us to answer questions such as where resources should be deployed to minimize the total area burned within a given time horizon, or if there are enough resources to protect a critical asset in the event of a fire.

We focus on the Wildfire Suppression Problem (WSP), a graph-based model first introduced by \textcite{Alvelos/2018}. After formally defining the WSP in Section~\ref{sec:problem_definition}, we provide a literature review of related problems, benchmark instances and algorithms in Section~\ref{sec:related_work}. In Section~\ref{sec:complexity}, we show that the WSP is strongly NP-complete even when all resources are released simultaneously, and also establish hardness on trees and full weighted directed grids. We further prove planar hardness for two related problems and grid hardness for one of them.

The only MIP formulation for WSP \parencite{Alvelos/2018}, has previously been found uncompetitive with tailored algorithms. In Section~\ref{sec:mip_formulation}, we propose a new MIP formulation that obtains state-of-the-art results. However, since existing benchmarks are small, easy to solve, and lack physical realism, Section~\ref{sec:instance_generator} introduces a new instance generator based on Rothermel's model of surface fire spread model~\parencite{Rothermel/1972} (summarized in Section~\ref{sec:rothermel}). While tailored to the problem we consider, the generator's principles are applicable to similar problems.

Finally, we use the generator in Section~\ref{sec:experiments} to evaluate all existing WSP algorithms. By varying instance characteristics, we identify scenarios where current state-of-the art algorithms struggle. We conclude and discuss future work in Section~\ref{sec:conclusions}.

\section{Problem Definition}\label{sec:problem_definition}

Let $G=(V,A)$ be a directed graph where non-negative arc weights $t_{uv}$, $uv\in A$, model travel times required for fire to propagate between vertices. For an ignition vertex $s \in V$, the travel times define a shortest-path tree rooted at $s$ in which each vertex $v \in V$ has a fire arrival time $a_v$. Now assume we have a set $R = [k]$ of suppression resources that can be allocated to a vertex $v \in V$, adding a delay $\Delta$ to its outgoing arcs. (For positive $n$, we write $[n] = \{1, \dots, n\}$.) Resource $i \in R$ is released at time $t_i$ and can only be allocated to a vertex $v$ if $a_v \geq t_i$, i.e., if $v$ is not burned yet. Assuming that each vertex can receive at most one resource we are interested in minimizing the number of vertices that burn before time horizon $H$.

An allocation of resources to vertices can be represented by an injective function $\A: R \rightarrow V$, where the image $P^{\A} = \mathrm{im}(\A)$ denotes the set of \emph{protected} vertices. An allocation can be partial; if $|P^{\A}| = |R|$ it is \emph{complete}. Let $\A_0$ denote the empty allocation that does not protect any vertex ($P^{\A_0} = \emptyset$). By definition, an allocation changes the travel times on the arcs, but it can also change the topology and the arrival times of the shortest-path tree. For allocation $\A$ let $t^{\A}$ denote the resulting fire propagation times, and $a^{\A}$ the fire arrival times. Let $B^{\A}_t = \{v \in V \mid a^{\A}_{v} < t\}$ be the set of vertices that are burned at time $t$. The problem is to find a feasible allocation $\A$ that minimizes the number of burned vertices at time $H$, $|B^{\A}_H|$.

Figure~\ref{fig:graph_models:problem_example} shows an example WSP instance on a nine-vertex graph with ignition vertex $s$. We have $k=3$ resources with release times $t_1 = 2$, $t_2 = 3$, and $t_3 = 4$. The right-hand side shows the allocation $\A=\{1\mapsto v_2, 2\mapsto v_4, 3\mapsto v_6\}$, protecting vertices $P^{\A} = \{v_2, v_4, v_6\}$. Note that vertices $v_1$ and $v_3$ cannot be protected as they burn at time $1$, before the first resource is released at time $2$. The allocation adds a delay $\Delta = 2$ to the outgoing arcs of $v_2$, $v_4$, and $v_6$, increasing arrival times at $v_5$, $v_7$, and $v_8$. For a time horizon $H = 5$ the objective value is $6$ since vertices $B^{\A}_{5} = \{s, v_1, v_2, v_3, v_4, v_6\}$ burn before time $5$.

\begin{figure}[]
  \centering
  \includegraphics[scale=0.8]{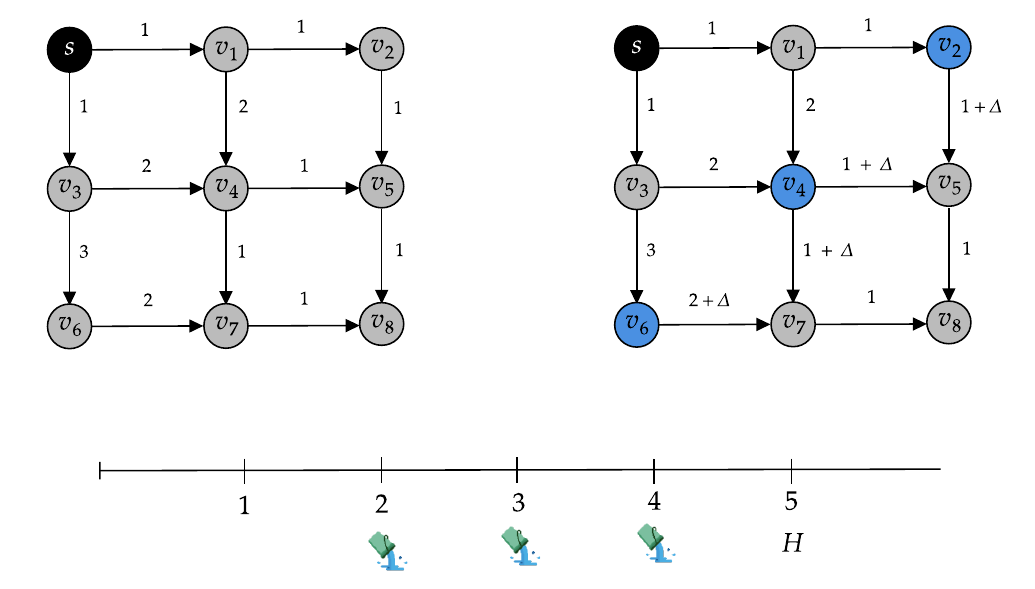}
  \caption[WSP instance example.]{WSP instance example with ignition vertex $s$ and $k=3$ resources released at times $t_1 = 2$, $t_2 = 3$, and $t_3 = 4$. With delay $\Delta=2$ and time horizon $H=5$ the right-hand side shows the optimal allocation $\A=\{1\mapsto v_2, 2\mapsto v_4, 3\mapsto v_6\}$.} \label{fig:graph_models:problem_example}
\end{figure}

\subsection{Notation} \label{sec:notation}

For any vertex $v \in V$, let $N^{+}_v$ and $N^{-}_v$ denote the set of vertices reachable via outgoing and incoming arcs, respectively. In Figure~\ref{fig:graph_models:problem_example}, we have $N^{+}_{v_4} = \{v_5, v_7\}$ and $N^{-}_{v_4} = \{v_1, v_3\}$. We represent times where resources become available by a sequence $t_1 \leq t_2 \leq \cdots \leq t_T$. Note that here $t_i$ denotes the $i$th time instant when one or more resources are released, rather than the release time of a specific resource $i \in R$. The context will make the meaning of $t_i$ clear. For each time $t \in [0, H]$, $R_t \subseteq R$ is the set of resources released at time $t$. In Figure~\ref{fig:graph_models:problem_example}, we have $t_1 = 2, t_2 = 3, t_3 = 4$, $R_1 = \emptyset$, $R_2 = \{1\}$, $R_3 = \{2\}$, and $R_4 = \{3\}$. For an arc $(u,v)\in A$ whose endpoints are denoted by single symbols we write $uv$. We further denote the shortest-path distance from $u$ to $v$ in $G$ as $d_{G}(u, v)$. For any subset $S \subseteq V$, $G \setminus S$ is the sub-graph induced by $V\setminus S$.

\section{Related Work}\label{sec:related_work}

Integrating spatial fire spread with optimization can be traced back to the linear programming model of \textcite{Hof/2000}. They represent the landscape as a grid of cells where fire propagation speed depends on the available fuel load. Suppression decisions determine the amount of fuel removed from each cell under a fixed suppression budget. Given an ignition cell, the objective is to maximize the fire arrival time to a target cell, representing an area to be protected from the fire. While pioneering, the model has several simplifying assumptions. First, it assumes that resources are available from the beginning, neglecting release times. Second, it assigns decision variables to fire entry and exit times of each cell, with their difference governed by a function of the fuel load. As we show in Appendix~\ref{app:related_work}, this implies uniform fire travel times for all outgoing arcs of a cell, preventing the model from capturing the directional effects of wind and slope.

\textcite{Wei/2011} extend the work of \textcite{Hof/2000} by allowing fire travel times of outgoing arcs depend on direction, and thus can model the influence of wind and slope. Similar to the WSP, their MIP formulation applies a uniform delay $\Delta$ to the outgoing arcs of a vertex upon resource allocation. They also consider a time horizon $H$ and are interested in the vertices that burn within this time horizon. However, different from the WSP, vertices $v \in V$ have a value $w_v$, and the objective is to minimize the total value of burned vertices.

\textcite{Wei/2011} also introduced safety constraints, by prohibiting the allocation of resources to cells where predicted flame lengths exceed a predefined threshold. However, like \textcite{Hof/2000} they assume that all resources are available from the outset. The authors address this limitation by proposing heuristic methods to distribute the suppression effort across multiple periods. We present the model of \textcite{Wei/2011} in our notation in Appendix~\ref{app:related_work}.

While these earlier models assume that resources are immediately available, \textcite{Alvelos/2018} was the first to add resource release times to a graph-based suppression model. This addition was part of a broader set of contributions. \textcite{Alvelos/2018} applied linear programming duality to the shortest-path problem to formulate MIP constraints for fire spread. This approach transforms fire arrival time calculations into a feasibility problem, ensuring that the resulting arrival times are accurate regardless of the objective function. By incorporating the resource release times, the author presents MIP models for four distinct wildfire suppression problems: (i) protecting priority areas (similar to \textcite{Hof/2000}), (ii) minimizing the total burned area within a time horizon (similar to \textcite{Wei/2011}), (iii) containing the fire by eliminating new ignitions, and (iv) containing the fire through covering the fire perimeter.

While previous research focused mainly on mathematical formulation and case studies, \textcite{Mendes/2022} is the first study dealing with algorithmic aspects of the WSP. They propose a set of instances, show that a commercial MIP solver can solve the formulation of \textcite{Alvelos/2018} only for small instances, and propose an Iterated Local Search (ILS) that obtains a better performance. The main components of the ILS are a perturbation procedure that tries to escape local minima by moving resources from protected to unprotected vertices, and a local search that promotes the formation of firebreak lines.

\textcite{Harris/2022} introduce the first exact algorithm for WSP, using logic-based Benders decomposition (LBBD). The algorithm uses the observation that, given a set of protected vertices, the problem reduces to a shortest-path calculation. Therefore, the WSP is split into a master problem, which decides time and location of resource deployment, and a subproblem that checks whether the decisions of the master problem are feasible. These problems are iteratively solved, and the master problem is updated with feasibility cuts after each iteration. The algorithm terminates once the optimal solution to the master problem is feasible for the original problem. They evaluate their algorithm on the instances from \textcite{Mendes/2022} consistently proving optimality within seconds. For further tests, they introduce $16$ new instances and compare the formulation of \textcite{Alvelos/2018}, the ILS of \textcite{Mendes/2022}, and LBBD. The results show that LBBD outperforms ILS, while the MIP formulation obtains the worst results, failing to find even feasible solutions for some instances.

\textcite{Delazeri/2024a} proposed an Iterative Beam Search (IBS) heuristic that explores a search tree of allocations, starting from the empty allocation $\A_0$. Tree levels correspond to release times $T$ and nodes to partial allocations. At each level, the algorithm expands partial allocations using the resources released at that time. The leaves represent complete allocations, and the best leaf is returned. The authors utilize two heuristic functions to rank partial allocations and prune the search tree, and a heuristic expansion rule that prioritizes vertices on the fire perimeter. Compared to LBBD, ILS, and the MIP formulation on the instances from \textcite{Harris/2022}, IBS consistently finds optimal solutions within seconds, while other algorithms need considerably more time.

\section{Computational Complexity} \label{sec:complexity}

The computational complexity of WSP has not been established in the literature. We show by a reduction from the Most Vital Nodes Problem (MVNP) that the WSP remains hard even when all resources are released simultaneously. This result is relevant for two reasons. First, it shows that relaxing the time constraints, a natural first approach for algorithm design, does not reduce the fundamental complexity. Second, it demonstrates that previous graph-based models which assume immediate resource availability are also NP-hard. Furthermore, the release schedule of the WSP lets us encode the \textsc{Firefighter} problem, yielding hardness even on trees and, by an embedding, on grids.

  The Weighted Wildfire Suppression Problem (WWSP) of \textcite{Wei/2011} and the Homogeneous Wildfire Suppression Problem (HWSP) of \textcite{Hof/2000} have no such temporal structure; for them we encode planar directed shortest-path edge interdiction. In summary, all three problems are strongly NP-complete on planar graphs. The WSP and WWSP remain so on full weighted directed grids, while the complexity of the HWSP on grids remains open.

Throughout this section, a \emph{full weighted directed grid} is a rectangular section of the integer lattice in which every pair of horizontally or vertically adjacent grid vertices is joined by an arc in each direction, with arbitrary positive rational travel times. We refer to the grid vertices as \emph{cells}. All problems considered below belong to NP: for a proposed allocation, its feasibility and all fire arrival times can be determined by a shortest-path computation.

\subsection{The WSP: temporal hardness}

In this section, we consider the decision version of the WSP that asks, given a bound $b_{\mathrm{WSP}}$, whether an allocation $\A$ exists such that the number of burned vertices $|B_H^{\A}|\leq b_{\mathrm{WSP}}$.

We first show that removing the scheduling choice does not make the problem easy: the WSP remains hard when every resource has the same release time.

\begin{definition}[Most Vital Nodes Problem, MVNP]
Given a directed graph $G=(V,A)$ with positive arc costs $t_{uv}$, distinct source and sink vertices $s,t\in V$, and positive integers $k$ and $h$, the MVNP asks whether there is a set $S\subseteq V\setminus\{s,t\}$ with $|S|\leq k$ such that $d_{G\setminus S}(s,t)\geq h$.
\end{definition}

The MVNP is strongly NP-complete, even when all arc costs are one~\parencite{BarNoy.etal/1995}.

\begin{theorem}\label{thm:wsp-simultaneous}
  The WSP is strongly NP-complete even when all resources are released simultaneously.
\end{theorem}
\begin{proof}
  Consider an MVNP instance $(G,t_{uv},s,t,k,h)$ and write $n = |V|$. Let $\mathcal N = \{v\in V:vt\in A\text{ and }d_G(s,v)+t_{vt}<h\}$ be the predecessors of sink $t$ that lie on an $s$--$t$ path shorter than $h$. For each $v\in\mathcal N$, create a set $U_v = \{u_1^v,\ldots,u_n^v\}$ of $n$ new leaves and replace the arc $(v,t)$ by the arcs $(v,u_i^v)$, each with travel time $t_{vt}$. Remove $t$ and all its incident arcs, retaining all other vertices, arcs, and travel times. Denote the resulting graph by $\widehat G$ and put $\mathcal U = \bigcup_{v\in\mathcal N}U_v$.

Let $\epsilon$ be the minimum arc travel time in the original instance. Construct a WSP instance on $\widehat G$ with ignition $s$, $k$ resources all released at the common time $\epsilon/2$, delay $\Delta = h$, horizon $H = h$, and burning bound $b_{\mathrm{WSP}} = n-1$. Since all travel times are positive, every reachable vertex other than $s$ has arrival at least $\epsilon$, so every vertex eligible in the MVNP solution is available when the resources are released; the positive release time merely prevents allocating a resource to the source.

Suppose $S$ is an MVNP solution. Allocate resources to the vertices of $S$. If some $u_i^v\in\mathcal U$ burned before $H$, there would be an $s$--$u_i^v$ path avoiding $S$ whose length, together with $t_{vt}$, is below $h$. Replacing its final arc $(v,u_i^v)$ by $(v,t)$ would give an $s$--$t$ path of the same length in $G\setminus S$, a contradiction. Thus no vertex of $\mathcal U$ burns, and at most the $n-1$ original vertices of $\widehat G$ burn. The WSP instance is therefore a yes-instance.

Conversely, suppose an allocation burns at most $n-1$ vertices, and let $S = P^{\A}\cap V$ be the protected original vertices. Then $|S|\leq k$ and $s,t\notin S$. If $G\setminus S$ contained an $s$--$t$ path of length below $h$, let $v$ be the predecessor of $t$ on that path. Then $v\in\mathcal N$. The corresponding path to $v$ in $\widehat G$ avoids every protected original vertex, and extending it to each of the $n$ leaves in $U_v$ gives an undelayed path of length below $h$. Allocating a resource to a leaf does not alter its own arrival time. Hence all $n$ vertices of $U_v$ burn, contradicting the bound $n-1$. Therefore $d_{G\setminus S}(s,t)\geq h$, so $S$ is an MVNP solution.

The construction has at most $n^2$ new vertices and uses only polynomially bounded numbers. It is therefore a strong polynomial reduction.
\end{proof}

We next strengthen the graph restriction to planar graphs. Unlike the preceding reduction, the hardness now comes from the temporal structure of the WSP. We reduce from \textsc{Firefighter}. Given a rooted tree $T$, the fire starts at its root $r$ at time zero. At each step $i=1,2,\ldots$, a firefighter defends an undefended, unburned vertex, permanently defending it and its descendants. The fire then advances from each burning vertex to its undefended neighbors. Defended or burned vertices remain permanently so. The decision problem asks whether at least $p$ vertices can be saved. It is NP-complete even on trees of maximum degree three~\parencite{Finbow/2007}.

A set $D$ of defended vertices may be assumed to be an antichain. It is schedulable when its vertices can be assigned distinct steps such that each $v\in D$ is defended no later than its depth $d(v)$. If $T_v$ is the subtree rooted at $v$, the number of saved vertices of set $D$ is $S(D) = \sum_{v\in D}|T_v|$.

\begin{theorem}\label{thm:wsp-planar}
  The WSP is strongly NP-complete on trees, and hence on planar graphs.
\end{theorem}
\begin{proof}
Let $(T,p)$ be a \textsc{Firefighter} instance, let $\delta$ be the depth of $T$, and orient every edge away from $r$. Give every arc unit travel time and use $r$ as the ignition vertex. Attach $M = \delta+1$ new out-neighbours to every vertex of $T$; denote the resulting directed tree by $\widehat T$. Use $k = \delta$ resources, and release one resource at each time $1,\ldots,\delta$. Set horizon $H = \delta + 2$, and delay $\Delta = H$.

Every original vertex $v$ has free-flow arrival $d(v)$. Thus assigning the resource released at time $i$ to $v$ is feasible exactly when $i\leq d(v)$, reproducing the \textsc{Firefighter} schedule. Since $\Delta=H$, protecting $v$ prevents fire from reaching every proper descendant of $v$ before the horizon. Protecting one of the new leaves has no effect, because WSP resources delay only outgoing arcs, so an allocation may be assumed to protect a schedulable antichain $D\subseteq V(T)$.

Each original vertex carries $M$ leaves and the subtrees rooted at vertices of $D$ are disjoint. The number of vertices saved by the WSP allocation is therefore
\[
 \sigma(D) = (1+M)S(D)-|D|.
\]
The subtraction accounts for the protected vertices themselves: a WSP resource delays their outgoing arcs but does not prevent those vertices from burning. Set the burning bound to
\[
 b_{\mathrm{WSP}} = |V(\widehat T)|-\big((1+M)p-k\big).
\]
By the choice of $b_{\mathrm{WSP}}$, an allocation burns at most $b_{\mathrm{WSP}}$ vertices if and only if $\sigma(D)\geq(1+M)p-k$. Since $0\leq|D|\leq k$,
\[
 (1+M)S(D)-k\;\leq\;\sigma(D)\;\leq\;(1+M)S(D).
\]
If \textsc{Firefighter} saves at least $p$ vertices with a defended set $D$, the first inequality gives $\sigma(D)\geq(1+M)p-k$, so the WSP instance is a yes-instance. Conversely, if $\sigma(D)\geq(1+M)p-k$, the second inequality gives $(1+M)S(D)\geq(1+M)p-k$, hence, using $1+M=\delta+2=k+2$,
\[
 S(D)\geq p-\frac{k}{k+2}>p-1,
\]
and integrality implies $S(D)\geq p$. The reduction is polynomial and all numbers are polynomially bounded, proving strong NP-completeness.
\end{proof}

The preceding reduction uses the successive resource releases essentially: if all resources are available at time zero, protecting the source would delay every outgoing arc beyond the horizon, making the constructed instances trivial. We next show that the temporal hardness survives the grid structure used in our application.

\begin{proposition}[Grid layout]\label{lem:grid-layout}
For a rooted tree $T$ on $n$ vertices of maximum degree three and an integer $G\geq1$, one can construct, in polynomial time, a lattice section of size $O(n\sqrt G)\times O(n\sqrt G)$ containing: a distinct cell for each $v\in V(T)$; for each edge $uv\in T$, a lattice path, called a \emph{channel}, joining the cells of $u$ and $v$; and, for each $v$, a set $C_v$ of exactly $G$ cells, the \emph{block} of $v$, admitting a Hamiltonian path that begins at a cell adjacent to the cell of $v$. The channels are pairwise grid-vertex-disjoint except at cells representing common endpoints. The sets $C_v$ are pairwise disjoint and disjoint from all channels. The total number of channel cells is $O(n^2\sqrt G)$.
\end{proposition}
\begin{proof}
For every $v\in V(T)$, attach a new leaf $z_v$ by the edge $vz_v$. This auxiliary edge serves only to reserve space for $C_v$. Since $T$ has maximum degree three, the augmented tree has maximum degree at most four, so it admits an orthogonal drawing on an $O(n)\times O(n)$ grid, with vertices represented by cells and edges represented by non-crossing lattice paths with at most two bends~\parencite{Biedl.Kant/1998}. Scale this drawing by $\lceil\sqrt G\rceil$. Each original tree edge then becomes a channel of length $O(n\sqrt G)$. For each $v$, the region associated with the auxiliary edge $vz_v$ contains enough unused grid cells to form a lattice path of exactly $G$ cells whose first cell is adjacent to the cell representing $v$; let $C_v$ be the set of cells on this path, which we call the \emph{block} of $v$; its cells are arranged as a snake, so fire entering the first cell reaches all $G$ cells in turn. The channels are disjoint except at shared tree vertices, and the reserved regions are pairwise disjoint and disjoint from the channels. Since $T$ has $n-1$ edges, the channels contain $O(n^2\sqrt G)$ cells in total.
\end{proof}

\begin{theorem}\label{thm:wsp-grid}
  The WSP is strongly NP-complete on full weighted directed grids.
\end{theorem}
\begin{proof}
  Reduce again from \textsc{Firefighter} on a rooted tree $T$ of maximum degree $3$, with $n = |V(T)|$, depth $\delta$, and target $p$. Apply Proposition~\ref{lem:grid-layout} with $G = n^5$. Direct every channel from parent to child and assign positive arc weights summing to one along each channel. Direct the block of each $v$ away from $v$, and assign positive weights to its entry and internal arcs that sum to less than one. Give every other grid arc weight at least $H = \delta+2$. Use ignition $r$, delay $\Delta = H$, one resource released at each time $1,\ldots,\delta$, and burning bound
\[
 b_{\mathrm{WSP}} = G(n-p+1)-1.
\]
The cell representing $v$ has arrival $d(v)$; the cells inside the channel entering $v$ have arrivals in $(d(v)-1,d(v))$; and the cells of $C_v$ have arrivals in $(d(v),d(v)+1)$. Every cell outside the channels and blocks has arrival at least $H$.

For the purpose of counting saved blocks, a resource inside the channel entering $v$ can be moved to $v$: the same release remains feasible and the same descendant blocks are saved. A resource inside $C_v$ can likewise be moved to $v$, saving at least the same blocks. Resources outside the low-weight structure have no effect before $H$. After removing duplicate and descendant assignments, the blocks saved by any allocation are therefore contained in those saved by a \textsc{Firefighter}-schedulable antichain $D$. Conversely, protecting the tree-vertex cells of any such $D$ is feasible and saves exactly the blocks corresponding to the $S(D)$ tree vertices saved in \textsc{Firefighter}.

Under this canonical allocation, $G(n-S(D))$ block cells burn. Every other cell that can burn is a channel cell or one of the $n$ tree-vertex cells, so their number is $O(n^2\sqrt G)+n=O(n^{4.5})+n<G$ for sufficiently large $n$; a fixed scaling of $G$ handles the remaining finite cases. Hence
\[
 |B_H^{\A}|=G(n-S(D))+q,\qquad 0\leq q<G.
\]
If $S(D)\geq p$, then $|B_H^{\A}|\leq G(n-p)+(G-1)=b_{\mathrm{WSP}}$. If every \textsc{Firefighter} defense saves at most $p-1$ tree vertices, then the preceding resource-relocation argument shows that any WSP allocation saves cells in at most $p-1$ blocks. Consequently, at least $n-p+1$ complete blocks burn, accounting for at least $G(n-p+1)>b_{\mathrm{WSP}}$ burned cells. Thus the two instances have the same answer. The construction is polynomial with polynomially bounded weights, proving strong NP-completeness.
\end{proof}

\subsection{The WWSP: planar shortest-path interdiction}

The weighted variant introduced by \textcite{Wei/2011} assigns values to vertices, assumes all resources are available at time zero, and allows vertices to be declared unavailable for protection. Since certain regions have higher priorities (e.g.~those containing critical infrastructure or residential areas), it is natural to assign a weight to each vertex and minimize the total lost weight.

\begin{definition}[Weighted Wildfire Suppression Problem, WWSP]
Given a directed graph $G=(V,A)$ with positive arc travel times, vertex values $w_v\geq0$, an ignition vertex $s\in V$, $k$ resources available at time zero, a forbidden set $F\subseteq V$, a common delay $\Delta \geq 0$, and a time horizon $H \geq 0$, the WWSP asks for an allocation $\A$ with $P^{\A}\cap F=\emptyset$ minimizing the total weight of the burned vertices
\[
 \sum_{v\in V}w_v[a_v^{\A}<H].
\]
Its decision version asks whether this value is at most a given bound $b_{\mathrm{WWSP}}$.
\end{definition}

We reduce from the directed shortest-path edge-interdiction problem (DSPEIP). An instance consists of a directed graph $G=(V,A)$, nonnegative integer arc lengths $\omega(a)$, nonnegative integer interdiction costs $c(a)$, vertices $u,v$, a budget $B$, and a threshold $L$. It asks whether there is $I\subseteq A$ with $c(I) = \sum_{a\in I} c(a) \leq B$ such that $d_{G\setminus I}(u,v)\geq L$. DSPEIP is strongly NP-complete on planar directed graphs~\parencite[Theorem~1.3]{Pan.Schild/2016}.

We first normalize an instance. Arcs of interdiction cost zero may be deleted. For $n = |V|$, set
\[
 \bar\omega(a) = 2n\omega(a)+2,
 \qquad \bar L = 2nL,
 \qquad \varepsilon = \min_{a\in A}\bar\omega(a).
\]
Every simple path $P$ has at most $n-1$ arcs, and hence, for every interdiction set $I$,
\begin{equation}\label{eq:dspeip-scaling}
 d_{G\setminus I}(u,v)\geq L
 \quad\Longleftrightarrow\quad
 \bar d_{G\setminus I}(u,v)\geq\bar L.
\end{equation}
Indeed, a path of original length at most $L-1$ has scaled length at most $2n(L-1)+2(n-1)<2nL$, whereas a path of original length at least $L$ has scaled length at least $2nL$.

Construct a directed graph $\widehat G$ by replacing each arc $a=xy$ by $c(a)$ internally disjoint two-arc paths
\[
 x\longrightarrow q_a^j\longrightarrow y,
 \qquad j=1,\ldots,c(a),
\]
introducing the auxiliary vertex $q_a^j$ on the $j$th path. Give the two arcs respective travel times $\varepsilon/2$ and $\bar\omega(a)-\varepsilon/2$. These times are positive and every path has total length $\bar\omega(a)$. The construction is planar: the parallel paths can be drawn in a narrow corridor around the embedding of the original arc. Because DSPEIP is strongly NP-complete, the costs and therefore the number of auxiliary vertices are polynomially bounded.

\begin{theorem}\label{thm:wwsp-planar}
  The WWSP is strongly NP-complete on planar graphs.
\end{theorem}
\begin{proof}
Use $\widehat G$ with ignition $s = u$, forbidden set $F = V$, $k = B$, horizon and delay $H = \Delta = \bar L$, target value $w_v = 1$, all other vertex values zero, and decision bound $b_{\mathrm{WWSP}} = 0$. Thus only auxiliary vertices $q_a^j$ can be protected.

If $I$ is a feasible DSPEIP interdiction, protect every auxiliary vertex associated with each arc $a\in I$. This uses $\sum_{a\in I}c(a)\leq B$ resources. Every $u$--$v$ path either traverses a protected auxiliary vertex, and therefore has length at least $\Delta=H$, or projects onto a path in $G\setminus I$, whose scaled length is at least $\bar L=H$ by \eqref{eq:dspeip-scaling}. Hence $a_v^{\A}\geq H$ and the WWSP objective is zero.

Conversely, let $\A$ have objective zero and define
\[
 I(\A) = \left\{a\in A:q_a^j\in P^{\A}\text{ for every }j=1,\ldots,c(a)\right\}.
\]
The sets of auxiliary vertices associated with distinct arcs are disjoint, so $c(I(\A))\leq|P^{\A}|\leq B$. Suppose that the original graph $G\setminus I(\A)$ contained a $u$--$v$ path $Q$ of length below $L$. For every arc $a$ of $Q$, we have $a\notin I(\A)$; hence, by the definition of $I(\A)$, there exists an index $j(a)$ such that $q_a^{j(a)}\notin P^{\A}$. Replacing every arc $a=xy$ of $Q$ by the path $x\to q_a^{j(a)}\to y$ would yield an undelayed $u$--$v$ path in $\widehat G$ of length below $\bar L=H$, contradicting $a_v^{\A}\geq H$. Thus $I(\A)$ is a feasible DSPEIP interdiction. All numerical values are polynomially bounded, completing the strong reduction.
\end{proof}

The forbidden set also allows the planar construction to be transferred to a grid: all routing vertices introduced by the transfer can simply be made unavailable for protection.

\begin{proposition}[Degree reduction]\label{lem:wwsp-degree}
The planar instances produced in Theorem~\ref{thm:wwsp-planar} can be transformed into equivalent planar WWSP instances of maximum total degree three, where the total degree of a vertex is the sum of its in- and out-degrees.
\end{proposition}
\begin{proof}
Only the forbidden original vertices can have high degree. Replace each such vertex $x$ by a forbidden binary in-tree merging its incoming arcs into a central vertex $x^*$ and a forbidden binary out-tree distributing the outgoing arcs. The trees are embedded according to the cyclic order of the incident arcs, preserving planarity. The auxiliary vertices $q_a^j$ are left unchanged.

All arcs from an original vertex to an auxiliary vertex $q_a^j$ have the common weight $\varepsilon/2$. Subdivide the out-tree to equal depth and assign positive weights so that each root-to-leaf path has total weight $\varepsilon/2$. Give the in-tree arcs sufficiently small positive rational weights that the total perturbation of any simple path is below one. The scaling in \eqref{eq:dspeip-scaling} leaves a gap of at least two around $\bar L$, so this perturbation cannot change the answer. Every introduced vertex is forbidden and every vertex has total degree at most three.
\end{proof}

\begin{proposition}[Grid embedding]\label{lem:wwsp-grid-embed}
A WWSP instance on a planar directed graph of maximum degree at most four can be transformed into an equivalent instance on a full weighted directed grid.
\end{proposition}
\begin{proof}
Take a planar orthogonal grid drawing~\parencite{Biedl.Kant/1998}, representing each arc by a non-crossing lattice channel. Split the arc's positive travel time among the forward channel arcs, give all other grid arcs weight at least $H$, and declare every channel-interior and unused cell forbidden. A cell representing a forbidden original vertex remains forbidden, while a cell representing a protectable vertex remains protectable. Protecting such a cell delays all its outgoing channels by $\Delta$, exactly reproducing protection in the planar instance. All other cells either cannot be protected or cannot be reached before $H$, so allocations and objective values correspond.
\end{proof}

\begin{theorem}\label{thm:wwsp-grid}
  The WWSP is strongly NP-complete on full weighted directed grids.
\end{theorem}
\begin{proof}
Apply Propositions~\ref{lem:wwsp-degree} and~\ref{lem:wwsp-grid-embed} to the hard planar instances of Theorem~\ref{thm:wwsp-planar}. Both transformations are polynomial and use only polynomially bounded rational weights.
\end{proof}

\subsection{The HWSP: homogeneous propagation}

The model of \textcite{Hof/2000} has vertex-dependent suppression delays but requires homogeneous propagation from each vertex. The latter precludes modeling directional influences like wind. Identical to \textcite{Wei/2011} all resources are released at time zero. The objective is to maximize the earliest fire arrival time to a given set of target vertices.

\begin{definition}[Homogeneous Wildfire Suppression Problem, HWSP]\label{def:hwsp}
Given a directed graph $G=(V,A)$ with positive travel times satisfying $t_{uv}=t_{uw}$ for all $v,w\in N^{+}_u$, an ignition vertex $s$, a target set $D$, $k$ resources available at time zero, and delays $\Delta_v\geq0$, the HWSP maximizes
\[
 \min_{v\in D}a_v^{\A}.
\]
Its decision version asks whether an allocation exists for which this value is at least a given threshold $b_{\mathrm{HWSP}}$.
\end{definition}

The graph $\widehat G$ used above already satisfies the homogeneity condition. Every outgoing arc of an original vertex has travel time $\varepsilon/2$, while every auxiliary vertex $q_a^j$ has only one outgoing arc. Vertex-dependent delays can now play the role of the WWSP forbidden set.

\begin{theorem}\label{thm:hwsp-planar}
  The HWSP is strongly NP-complete on planar graphs.
\end{theorem}
\begin{proof}
Starting from the normalized planar DSPEIP instance, use $\widehat G$, ignition $s = u$, target set $D = \{v\}$, $k = B$, and threshold $b_{\mathrm{HWSP}} = \bar L$. Define
\[
 \Delta_x =
 \begin{cases}
  \bar L,&x=q_a^j\text{ for some }a,j,\\
  0,&x\in V.
 \end{cases}
\]
Protecting an original vertex has no effect, so an allocation may be assumed to protect only auxiliary vertices $q_a^j$. If a DSPEIP solution $I$ exists, protect all auxiliary vertices associated with every $a\in I$. Exactly as in Theorem~\ref{thm:wwsp-planar}, every $u$--$v$ path then has length at least $\bar L$, and hence $a_v^{\A}\geq b_{\mathrm{HWSP}}$.

Conversely, from an allocation with $a_v^{\A}\geq\bar L$, take the arcs all of whose associated auxiliary vertices are protected. Their total interdiction cost is at most $B$. A path of original length below $L$ avoiding these arcs could be lifted through one unprotected auxiliary vertex per arc, producing an undelayed path of length below $\bar L$, a contradiction. Planarity, polynomial size, and polynomially bounded numbers follow from the WWSP construction.
\end{proof}

\section{A new Mixed-Integer Programming Formulation} \label{sec:mip_formulation}

In this section, we present a MIP formulation for the WSP. For any vertex $v \in V$, let the continuous variable $a_v \in \mathbb{R}_{\geq 0}$ denote the fire arrival time at $v$. The binary variable $y_v \in \{0, 1\}$ indicates whether $v$ burns before the horizon $H$. For each resource release point $i \in [T]$, the binary variable $r_{iv} \in \{0, 1\}$ determines whether vertex $v$ receives a resource released at time $t_i$. The MIP model is as follows.

\begin{align*}
  \text{min.} \quad & \sum_{v \in V} y_v \tag{M.1} \label{mg:obj} \\
  \text{s.t.} \quad
  & a_s = 0, &&  \tag{M.2} \label{mg:source} \\
  & a_v \leq a_u + t_{uv} + \Delta \sum_{t \in T} r_{tu}, &&  uv \in A, \tag{M.3} \label{mg:spread} \\
  & \sum_{v \in V} r_{iv} \leq |R_{t_i}|, &&  i \in [T], \tag{M.4} \label{mg:res_limited} \\
  & \sum_{i \in [T]} r_{iv} \leq 1, &&  v \in V, \tag{M.5} \label{mg:res_unique} \\
  & a_v \geq r_{iv} t_i, &&  i \in [T],\; v \in V, \tag{M.6} \label{mg:res_active} \\
  & y_v \geq 1 - a_v / H, &&  v \in V, \tag{M.7} \label{mg:burned_def} \\
  & y_v, r_{iv} \in \{0,1\}, &&  v \in V,\; i \in [T],  \\
  & a_v \in \mathbb{R}_{\geq 0}, &&  v \in V.
\end{align*}

Constraints \eqref{mg:source} and \eqref{mg:spread} define fire propagation using the structure of a shortest-path problem. The ignition time $a_s=0$ anchors the fire spread, while \eqref{mg:spread} enforces that the arrival time at any vertex $v$ cannot exceed the arrival time at a predecessor plus the travel and delay costs.  Because of the objective \eqref{mg:obj} and constraint \eqref{mg:burned_def}, fire arrival times $a_v$ are incentivized to increase (to allow $y_v$ to drop to zero).  Resource constraints are handled by \eqref{mg:res_limited}--\eqref{mg:res_active}. Constraint~\eqref{mg:res_limited} restricts the allocation based on available resources $|R_i|$ at each release time $t_i$, and \eqref{mg:res_unique} ensures each vertex is protected at most once. Constraint~\eqref{mg:res_active} enforces temporal feasibility: a resource can only be allocated to vertex $v$ if it is available at time $t_i$ before the fire arrives. Finally, \eqref{mg:burned_def} determines the burning status $y_v$ relative to the planning horizon $H$. The objective~\eqref{mg:obj} minimizes the total number of burned vertices.

\paragraph{A note on MIP models}

\textcite{Alvelos/2018} proposed the first MIP for the WSP, utilizing linear programming duality to enforce exact fire arrival times at each vertex. While mathematically rigorous, this approach is computationally demanding; previous work by \textcite{Harris/2022} and \textcite{Mendes/2022} report that the model is hard to solve even for small instances.

Prior to \textcite{Alvelos/2018}, earlier work on fire spread in graphs by \textcite{Hof/2000} and \textcite{Wei/2011} modeled propagation using simpler inequalities of the form $a_v \leq a_u + t_{uv}$ for each arc $uv \in A$. These constraints do not strictly enforce equality between $a_v$ and the shortest path distance. However, in these models, arrival times primarily support decision variables that track whether a vertex burns within the time horizon $H$. Because objective $ \text{min.}\; \sum_v y_v$ and constraints $y_v \geq 1 - a_v/H$ incentivize the model to increase $a_v$, the resulting burning status $y_v$ remains correct.

While these earlier works lacked the time-constrained suppression features introduced by \textcite{Alvelos/2018}, our formulation bridges these two approaches.  We combine efficient propagation constraints by \textcite{Hof/2000} and \textcite{Wei/2011} with the temporal resource constraints of \textcite{Alvelos/2018}. As demonstrated in the experimental section, the resulting formulation outperforms Alvelos model, and remains competitive, and in some cases superior, to state-of-the-art specialized algorithms.

\section{Rothermel's Model for Surface Fire Spread: A Primer} \label{sec:rothermel}

In 1971, \textcite{Frandsen/1971} proposed a theoretical model for fire spread based on the principle that fire propagates as a series of ignitions, i.e., unburned fuel is heated by the approaching fire front until it reaches ignition temperature. Figure~\ref{fig:rothermel} illustrates this process, where a unit volume of fuel $\Delta V$ (a representative volume of the fuel bed that retains the same properties of the surrounding fuel) is heated by an advancing fire. Using the principle of conservation of energy, \textcite{Frandsen/1971} derived the rate of spread (ROS) as the ratio between the energy released by the fire as it approaches the volume and the energy required to ignite it.

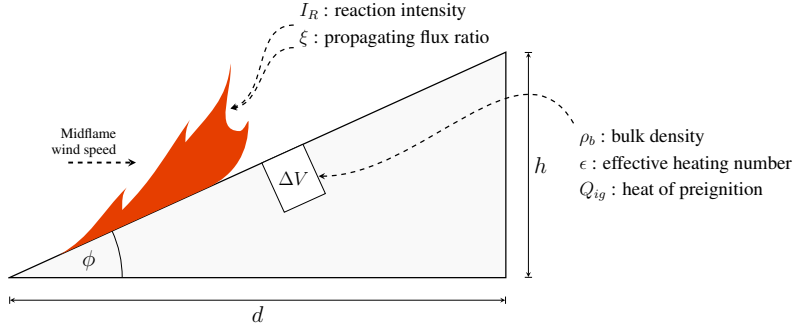
\begin{figure}[h]
  \centering
  \begin{tikzpicture}[
    >=stealth,
    thick,
    font=\small
]

\def\lenD{11} 
\def\lenH{5}  
\coordinate (A) at (0, 0);
\coordinate (B) at (\lenD, 0);
\coordinate (C) at (\lenD, \lenH);

\pgfmathsetmacro{\slopeangle}{atan(\lenH/\lenD)}

\draw[fill=gray!5] (A) -- (B) -- (C) -- cycle;

\draw[|<->|, thin] ($(A)-(0,0.5)$) -- node[below] {\Large $d$} ($(B)-(0,0.5)$);
\draw[|<->|, thin] ($(B)+(0.5,0)$) -- node[right] {\Large $h$} ($(C)+(0.5,0)$);

\draw[thin] ($(A)+(2.5,0)$) arc (0:\slopeangle:2.5);
\node at ($(A)+(\slopeangle/2:1.8)$) {\Large $\phi$};


\begin{scope}
    \definecolor{flameorange}{RGB}{255, 69, 0}
    \definecolor{flameyellow}{RGB}{255, 215, 0}

    \path[fill=flameorange!90!black, draw=none, very thick, line join=round]
        ($(A)!0.1!(C)$) coordinate (fstart)
        to[out=30, in=225] ++(1.5, 1.5) coordinate (tip1)
        to[out=235, in=90] ++(-0.1, -0.5)
        to[out=45, in=235] ++(1.5, 2.0) coordinate (tip2)
        to[out=245, in=90] ++(-0.1, -0.5)
        to[out=50, in=265] ++(1.0, 1.75) coordinate (tip3)
        to[out=275, in=180] ++(0.2,-1.5)
        to[out=0, in=90] ++(0.2,0.2)
        to[out=-90, in=25] ($(A)!0.4!(C)$)
        -- cycle;
\end{scope}

\node[anchor=east, align=right] (windLabel) at (2.5, 3) {Midflame\\wind speed};
\draw[dashed, ->, very thick] (windLabel.south west) ++(0.5,0) -- ++(1.5,0);

\coordinate (BoxAnchor) at ($(A)!0.55!(C)$);
\begin{scope}[rotate around={\slopeangle:(BoxAnchor)}]
    \draw[fill=white, thin] ($(BoxAnchor)-(0.5,1.2)$) rectangle ($(BoxAnchor)+(0.5,0)$);
    \node at ($(BoxAnchor)-(0,0.6)$) {\large $\Delta V$};
\end{scope}

\node[align=left, anchor=west] (topParams) at (\lenD-4.7, \lenH+0.6)
    {\large $I_R$ : reaction intensity \\[5pt] \large $\xi$ : propagating flux ratio};

\draw[dashed, ->] (topParams.west) to[out=180, in=20] ($(tip3)-(0,1)$);
\draw[dashed, ->] ($(topParams.west)-(0,0.5)$) to[out=180, in=10] ($(tip3)-(0,1)$);

\node[align=left, anchor=west] (bottomParams) at (\lenD+1.5, 2.5)
    {\large $\rho_b$ : bulk density \\[5pt]
     \large $\epsilon$ : effective heating number \\[5pt]
     \large $Q_{ig}$ : heat of preignition};

\draw[dashed, ->] (bottomParams.north west) to[out=120, in=0] ($(BoxAnchor)-(-0.8,0.5)$);

\end{tikzpicture}
  \caption[Rothermel's fire spread model.]{Rothermel's fire spread model. As the fire front advances, the unit volume $\Delta V$ is heated until it ignites.} \label{fig:rothermel}
\end{figure}

Frandsen's model contains terms that are difficult to derive from first principles, which hinders its applicability in real settings. In 1972, \textcite{Rothermel/1972} addressed this by decomposing Frandsen's equation into a few parameters that could be estimated through laboratory experiments using artificial fuel beds and wind tunnels. Rothermel's model is one of the most widely used tools in the prevention and combat of wildfires.

Rothermel's model defines the ROS (in ft/min) of a fire as
\begin{equation} \label{eq:background:rothermel}
  R = \frac{I_R \xi}{\rho_b \epsilon Q_{ig}} (1 + \Phi_w + \Phi_s).
\end{equation}
The variables in Equation~\ref{eq:background:rothermel} are determined by the characteristics of the fuel particles (i.e.~a discrete unit of combustible material, such as a leaf of grass or a piece of wood), the fuel bed (i.e.~a collection of fuel particles), and the environment. This equation is best understood as the product of two components: the base rate of spread under no-wind, no-slope conditions $R_0= \frac{I_R \xi }{\rho_b \epsilon Q_{ig}}$ (ft/min) and a dimensionless slope-and-wind correction factor $(1 + \Phi_w + \Phi_s)$. We can thus reformulate
\begin{equation}\label{eq:background:rothermelReformulation}
  R = R_0(1 + \Phi_w + \Phi_s).
\end{equation}

Rothermel's original model assumes that fire spreads in the direction of both wind and slope, a scenario known as \emph{upslope headfire}. More generally, a fire spreading in the same direction as the wind is a \emph{headfire}, otherwise, a \emph{backfire}. Similarly, fire moving uphill is \emph{upslope}, otherwise, \emph{downslope}. This gives four possible scenarios: upslope headfire, downslope headfire, upslope backfire, and downslope backfire.

To account for these scenarios, several extensions of Rothermel's model have been proposed. In this work, we adopt the formulation introduced by Albini and described by \textcite{Weise/1997}. Among the alternatives compared in their study, Albini's extension showed the best agreement with observed data. The rate of spread $R$ according to Albini's model is given by $R = R_0r$, where multiplier $r$ depends on the wind and slope scenario
\begin{equation}
r =
\left\{
\begin{array}{ll}
  1 + \Phi_w + \Phi_s,            & \quad \text{Upslope headfire;}   \\
  1 + \max\{0, \Phi_w - \Phi_s\}, & \quad \text{Downslope headfire;} \\
  1 + \max\{0, \Phi_s - \Phi_w\}, & \quad \text{Upslope backfire;}   \\
  1,                              & \quad \text{Downslope backfire.}
\end{array}
\right.
\end{equation}
In the following sections, we explain each component of Equation~\ref{eq:background:rothermelReformulation} in detail.

\subsection{No-wind, no-slope Rate of Spread}

The no-wind, no-slope rate of spread $R_0$ (ft/min) characterizes how quickly fire propagates through a flat terrain in the absence of wind. It is defined as $$R_0 = \frac{I_R \xi }{\rho_b \epsilon Q_{ig}}.$$

The \emph{reaction intensity} $I_R$ (Btu/ft$^2$/min) represents the rate at which energy is released per unit area of the fire front. This value depends on the amount of fuel available to burn and the physical and chemical properties of the fuel particles. For instance, a dense layer of fallen branches and logs burns with high intensity, releasing a large amount of energy in a short period, resulting in a high $I_R$. In contrast, a cover of dry grass burns more gradually, producing a lower $I_R$.

The \emph{propagating flux ratio} $\xi$ is a dimensionless coefficient that quantifies how much of the reaction intensity actually contributes to igniting the unburned fuel ahead. Not all the energy released by combustion is useful for fire spread, as some portion is lost to the surroundings. A fire in a continuous fuel bed may have a higher $\xi$ than one burning over a sparse cover of shrubs because more of the heat is effectively transferred forward.

The reaction intensity $I_R$ and the propagating flux ratio $\xi$ determine the energy available for combustion, while the remaining terms of Equation~\ref{eq:background:rothermel} control how much of this energy is necessary to ignite the fuel ahead of the fire.

The \emph{bulk density} $\rho_b$ (lb/ft$^3$) represents the amount of fuel mass per unit volume of the fuel bed. gA higher bulk density means more fuel is available to burn in a given area, which can slow down the spread of fire due to reduced airflow.

The \emph{effective heating number} $\epsilon$ is a dimensionless value that represents the proportion of each fuel particle that reaches ignition temperature during combustion. In fine fuels like grass, most of the fuel particle is quickly heated, so $\epsilon$ is close to 1. In contrast, thick branches heat more slowly, leading to a lower $\epsilon$.

The \emph{heat of preignition} $Q_{ig}$ (Btu/lb) is the amount of energy required to bring a unit mass of fuel to ignition, and depends on the moisture content of the fuel particles. Dry fuels, such as dead grass or wood with low moisture content, require less energy to ignite and thus have a lower $Q_{ig}$, while wet logs need significantly more energy before they start burning.

\subsection{Slope Factor}

The slope factor $\Phi_s$ is a dimensionless quantity that accounts for the influence of terrain steepness on fire spread. When fire burns on an inclined surface, the flames tilt towards the slope, bringing them closer to the unburned fuel, which increases heat transfer and accelerates the fire spread rate. This effect is captured by the equation
\begin{align*}
  \Phi_s(A;  \beta) &= C_s(\beta) A^2,
\end{align*}
where $A = \tan \phi$ is the tangent of the slope angle, $\beta$ is the packing ratio, and $C_s$ is a function of the packing ratio defined as $C_s(\beta) = \const{a_s} \beta^{-\const{b_s}}$. Constants $\const{a_s}$ and $\const{b_s}$ are defined in Table~\ref{tab:constants}.

The \emph{packing ratio} $\beta$ is a dimensionless measure of how densely fuel particles are arranged within the fuel bed. It quantifies the trade-off between fuel availability and airflow in a fuel bed, where higher values indicate more fuel per unit volume but reduced oxygen penetration, and lower values indicate greater airflow but less fuel to sustain combustion. A low packing ratio can be found in scattered shrubs or dry grass, while a high packing ratio can be found in compact litter.

\begin{table}
  \centering
  \caption{Numerical constants used in Rothermel's model.}
  \label{tab:constants}

  \smallskip
  \begin{tabular}{lrrrrrrrrr}
    \toprule
    Constant & $\const{a_s}$ & $\const{b_s}$ & $\const{a_w}$ & $\const{b_w}$ & $\const{c_w}$ & $\const{d_w}$ & $\const{e_w}$ & $\const{f_w}$ & $\const{g_w}$        \\
    Value    & $5.275$       & $0.3   $      & $7.47   $     & $0.133  $     & $0.55   $     & $0.02526$     & $0.54   $     & $0.715  $     & $3.59\times 10^{-4}$ \\
    \bottomrule
  \end{tabular}
\end{table}

\subsection{Wind Factor}

The wind factor $\Phi_w$ is a dimensionless quantity that quantifies how wind enhances the spread rate of fire. As wind speed increases, flames are tilted forward into the unburned fuel, preheating and igniting it more quickly. This effect is captured by the equation
$$\Phi_w(U; \sigma,  \beta_{rel}) = C_w(\sigma, \beta_{rel})U^{B_w(\sigma)},$$
where $U$ is the mid-flame wind speed (ft/min), representing the wind velocity at the level of the flames, and $\beta_{rel}$ and $\sigma$ are parameters that depend on the fuel type. Functions $C_w$ and $B_w$ are defined as
\begin{align*}
  C_w(\sigma, \beta_{rel}) &= (\const{a_w} e^{-\const{b_w} \sigma^{\const{c_w}}}) (\beta_{rel}^{-\const{d_w} e^{-\const{e_w} \sigma}}), \\
  B_w(\sigma) &=  \const{f_w} \sigma^{\const{g_w}},
\end{align*}
where $\const{a_w}$, $\const{b_w}$, $\const{c_w}$, $\const{d_w}$, $\const{e_w}$, $\const{f_w}$, and $\const{g_w}$ are constants defined in Table~\ref{tab:constants}.

The \emph{surface-area-to-volume ratio} $\sigma$ (ft$^2$/ft$^3$) characterizes the fineness of the fuel particles. Fine fuels such as dry grass have high values of $\sigma$, whereas logs and thick branches have much lower values. Earlier, we defined the packing ratio $\beta$ as a measure of how densely fuel particles are packed in the fuel bed. For any given fuel type, there exists an \emph{optimum packing ratio} $\beta_{op}$, which is the value of $\beta$ that maximizes the energy release for combustion. The \emph{relative packing ratio} $\beta_{rel} = \beta / \beta_{op}$ is the ratio between $\beta$ and $\beta_{op}$ and determines how efficiently a given fuel bed sustains fire spread.

\subsection{Fuel Models}

Rothermel's model has several parameters to estimate fire spread, such as reaction intensity, bulk density, and propagating flux ratio. These parameters can be measured empirically, but collecting data for every possible fire scenario is impractical. To address this, fire scientists classify vegetation into categories assign representative parameters to each. These predefined parameters are known as \emph{fuel models}. \textcite{Andrews/2018} compiled a list of $53$ fuel models that are widely used in fire behavior prediction. In the following section, we will utilize these models to estimate values for components such as the packing ratio $\beta$, the relative packing ratio $\beta_{rel}$, the surface-area-to-volume ratio $\sigma$, and the no-wind, no-slope rate of spread $R_0$.

\section{A New Instance Generator} \label{sec:instance_generator}

Existing instances for the WSP in the literature are no challenge for current state-of-the-art methods. For example, the beam search of \textcite{Delazeri/2024a}, finds the optimal solution for nearly all such instances in under five minutes; similarly, \textcite{Harris/2022} report that logic-based Benders decomposition can prove optimality for all existing instances within a few hours. Furthermore, these benchmarks lack physical realism, as their numerical parameters have no physical basis. Fire travel times, for instance, come without units, and are assigned randomly to each vertex, and thus are not spatially correlated. This makes it impossible to connect computational results to real-world suppression scenarios. Additionally, the literature lacks a standardized, reproducible method for generating new instances.

To address these limitations, we propose a new instance generator based on Rothermel's fire spread model, which provides a physically grounded calculation of fire travel times. This tool enables the creation of benchmarks with different grid sizes, delay values, and resource availability, while allowing for controlled variation of problem complexity. Table~\ref{tab:gfactors} summarizes the parameters and their default settings. They are detailed in the following sections.

\begin{table}[]
  \centering
  \caption[Parameters of the instance generator.]{Parameters of the instance generator. The generator supports eight experimental factors, each with a number of default levels. The table shows the values for each level, with the default value in \textbf{bold}. Parameter groups are: IS = Instance size; EF = Environmental factors; SC = Suppression capacity; RW = Release window.}
  \label{tab:gfactors}
  \small
  \setlength{\tabcolsep}{5pt} 
  \begin{tabular}{clcl clcl}
      \toprule
      Group & Factor                          & Level             & Value                    & Group & Factor                             & Level              & Value   \\
      \midrule
      \multirow{4}{*}{IS} & \multirow{4}{*}{Grid size $n$}    & Small            & $20$                     & \multirow{3}{*}{SC} & \multirow{3}{*}{Delay $\Delta$}    & Low                & $H/3$   \\
                          &                                   & \textbf{Medium}  & $30$                     &                      &                                      & Medium             & $H/2$   \\
                          &                                   & Large            & $40$                     &                      &                                      & \textbf{High}      & $H$     \\
                          &                                   & Huge             & $80$                     &                      &                                      &                    &         \\
      \midrule
      \multirow{3}{*}{IS} & \multirow{3}{*}{Decision points $T$} & Few            & $5$                      & \multirow{3}{*}{SC} & \multirow{3}{*}{Resources $k$}      & Few                & $n/2$   \\
                          &                                      & \textbf{Moderate} & $10$                   &                      &                                      & \textbf{Moderate}  & $n$     \\
                          &                                      & Many           & $20$                     &                      &                                      & Many               & $2n$    \\
      \midrule
      \multirow{3}{*}{EF} & \multirow{3}{*}{Slope $N_z$ (ft)}   & Flat           & $N_{xy}\tan(10^{\circ})$ & \multirow{3}{*}{RW} & \multirow{3}{*}{First Release Time} & \textbf{Early}     & $q(5)$  \\
                          &                                      & \textbf{Moderate} & $N_{xy}\tan(20^{\circ})$ &                     &                                       & Late               & $q(10)$ \\
                          &                                      & Steep          & $N_{xy}\tan(40^{\circ})$ &                      &                                       & Very late          & $q(20)$ \\
      \midrule
      \multirow{3}{*}{EF} & \multirow{3}{*}{Wind Speed (ft/min)} & \textbf{Light} & $[94.5,\,195.0]$         & \multirow{4}{*}{RW} & \multirow{4}{*}{Last Release Time}  & Very early         & $q(60)$ \\
                          &                                      & Moderate       & $[324.9,\,466.5]$        &                      &                                       & Early              & $q(70)$ \\
                          &                                      & Strong         & $[637.8,\,815.1]$        &                      &                                       & Late               & $q(80)$ \\
                          &                                      &                &                          &                      &                                       & \textbf{Very late} & $q(95)$ \\
      \bottomrule
   \end{tabular}
\end{table}

\subsection*{Landscapes}

A landscape is represented as a two-dimensional grid of square cells, each with an associated height. Figure~\ref{fig:instance:landscape:discretized} shows an example. We encode this landscape as a directed graph $G=(V, A)$, where each vertex corresponds to a cell, and arcs represent adjacencies between neighboring cells. Neighborhood relations are defined in the $xy$-plane, where each vertex can have up to four neighbors: two horizontal (left and right) and two vertical (top and bottom). For any vertex $v \in V$, we denote its spatial position as the tuple $(v_x, v_y, v_z) \in \mathbb{R}^3$, where $v_x$ and $v_y$ are the coordinates in the $xy$-plane and $v_z$ is the height. Figure~\ref{fig:instance:landscape:graph} illustrates this structure. The Euclidean distance between any two vertices $u, v \in V$ in the three-dimensional space is denoted by $d(u, v)$.

\begin{figure}[]
  \centering
  \begin{subfigure}{0.45\textwidth}
      \centering
      \includegraphics[width=\textwidth]{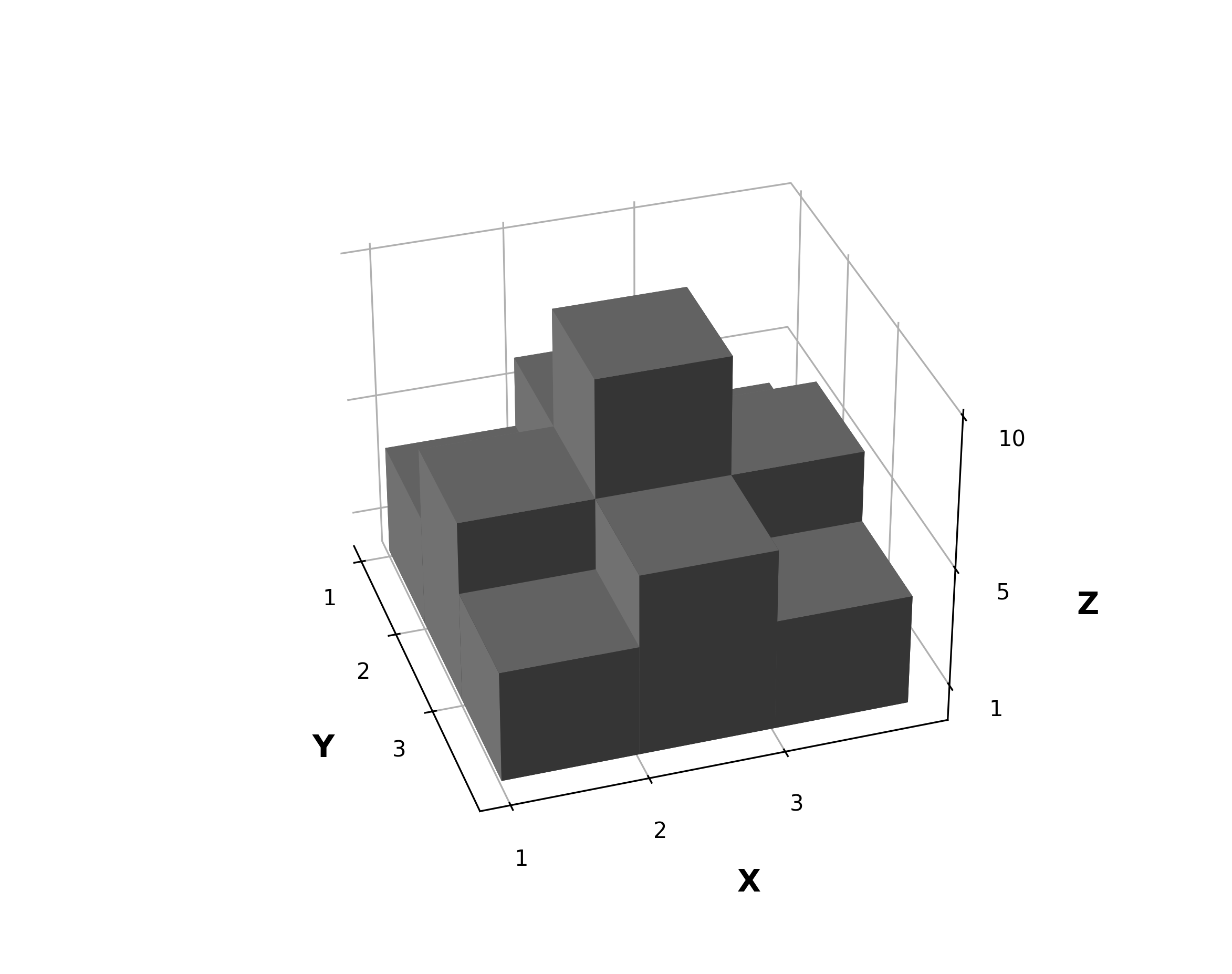}
      \caption{Landscape.}
      \label{fig:instance:landscape:discretized}
  \end{subfigure}
  \hspace{0.5cm}
  \begin{subfigure}{0.45\textwidth}
      \centering
      \includegraphics[width=\textwidth]{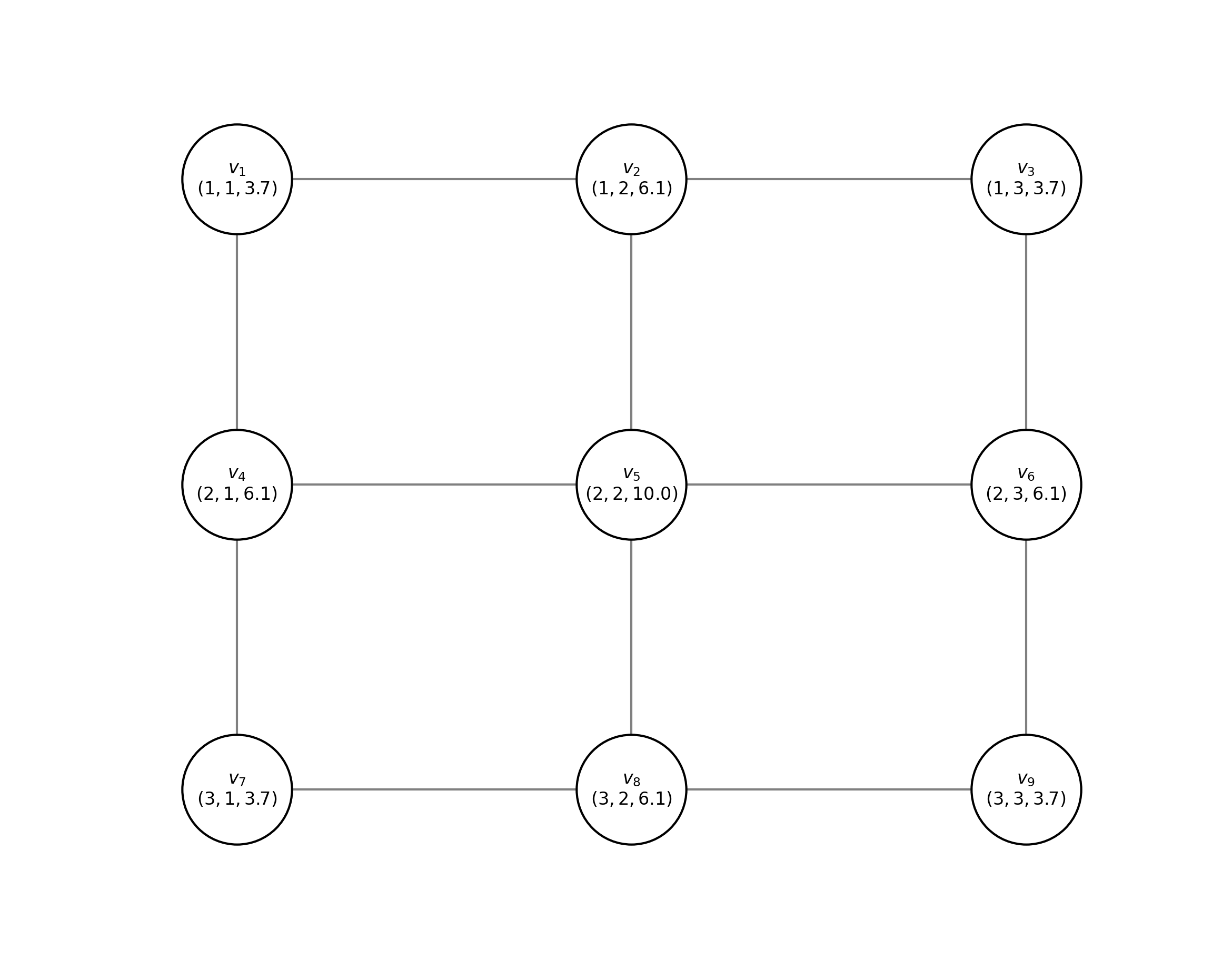}
      \caption{Graph representation.}
      \label{fig:instance:landscape:graph}
  \end{subfigure}
  \caption[Visualization of a discretized $3 \times 3$ landscape and its graph representation.]{Visualization of a discretized $3 \times 3$ landscape (left) and its graph representation (right). In the example, $N_{xy} = 3$, $N_{z} = 10$, and $d = 1$. The graph representation shows the coordinates of the vertices.}
  \label{fig:instances:landscape}
\end{figure}

We define four different grid sizes $n \times n$, where $n$ is the number of cells in each dimension in the $xy$-plane (see Table~\ref{tab:gfactors}). In all instances, the fire starts from a single ignition vertex located at the center of the grid.

The distance between two adjacent cells in the $xy$-plane is denoted by $d$. We assume that each grid represents a square landscape of approximately $N_{xy} \times N_{xy}$ square feet, where $N_{xy}$ is a parameter of the instance generation process. Given a grid size $n \times n$, the distance $d$ between two adjacent cells is defined as $d = \lceil N_{xy} / n \rceil$. Since all landscapes cover approximately the same area, the grid size determines the resolution of the landscape.

\begin{figure}[]
  \centering

  \begin{subfigure}[t]{.45\textwidth}
    \centering
    \includegraphics[width=\linewidth]{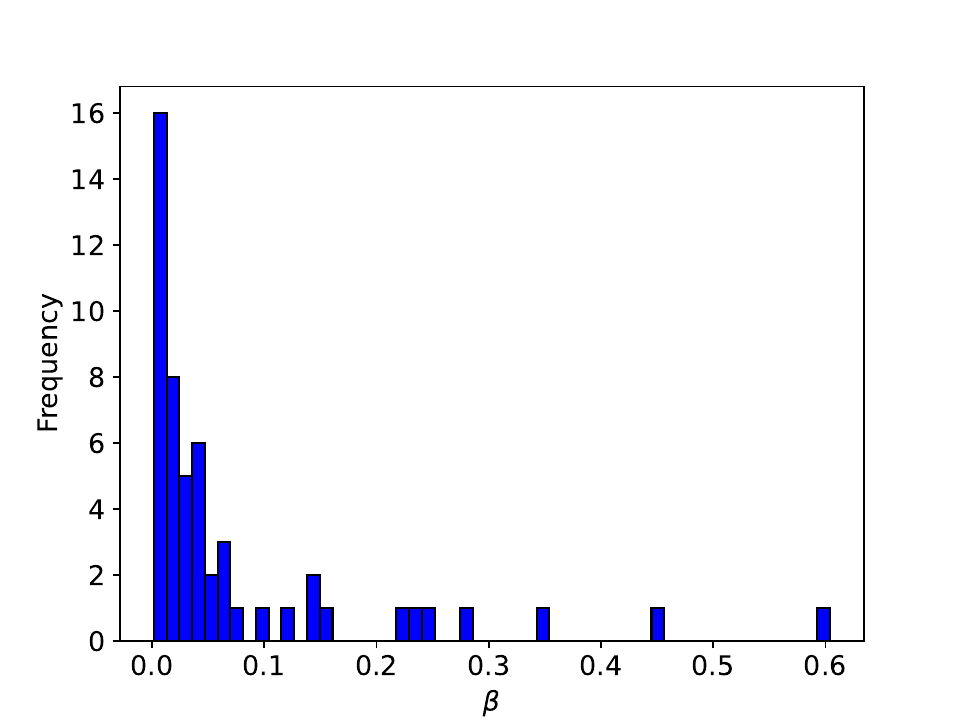}
    \caption{Packing ratio.}
    \label{fig:beta}
  \end{subfigure}\hfill
  \begin{subfigure}[t]{.45\textwidth}
    \centering
    \includegraphics[width=\linewidth]{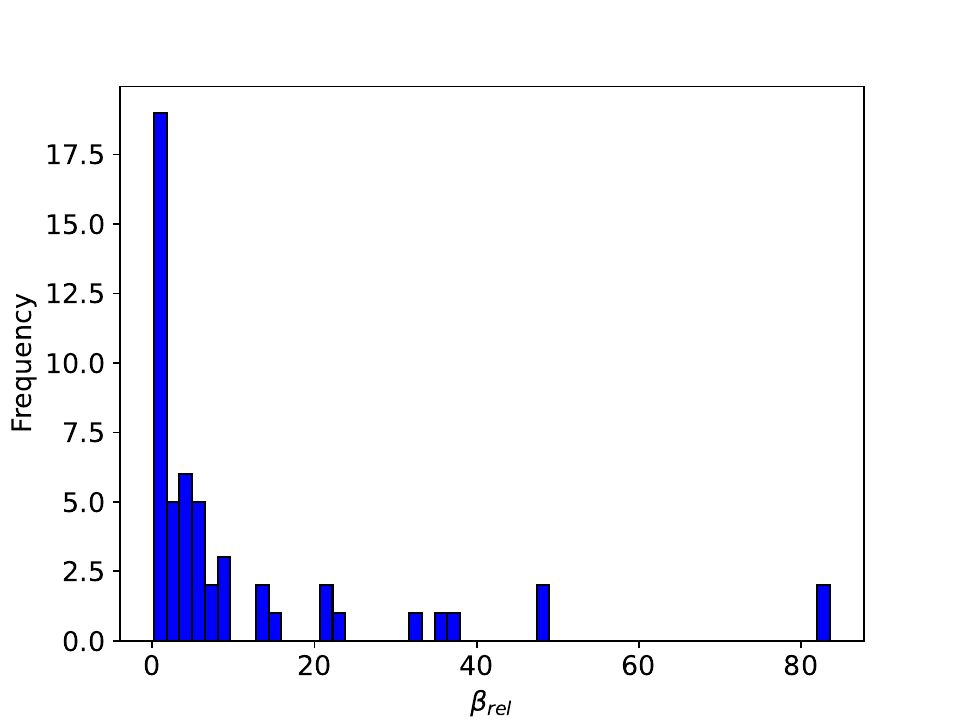}
    \caption{Relative packing ratio.}
    \label{fig:beta_rel}
  \end{subfigure}

  \medskip            

  \begin{subfigure}[t]{.45\textwidth}
    \centering
    \includegraphics[width=\linewidth]{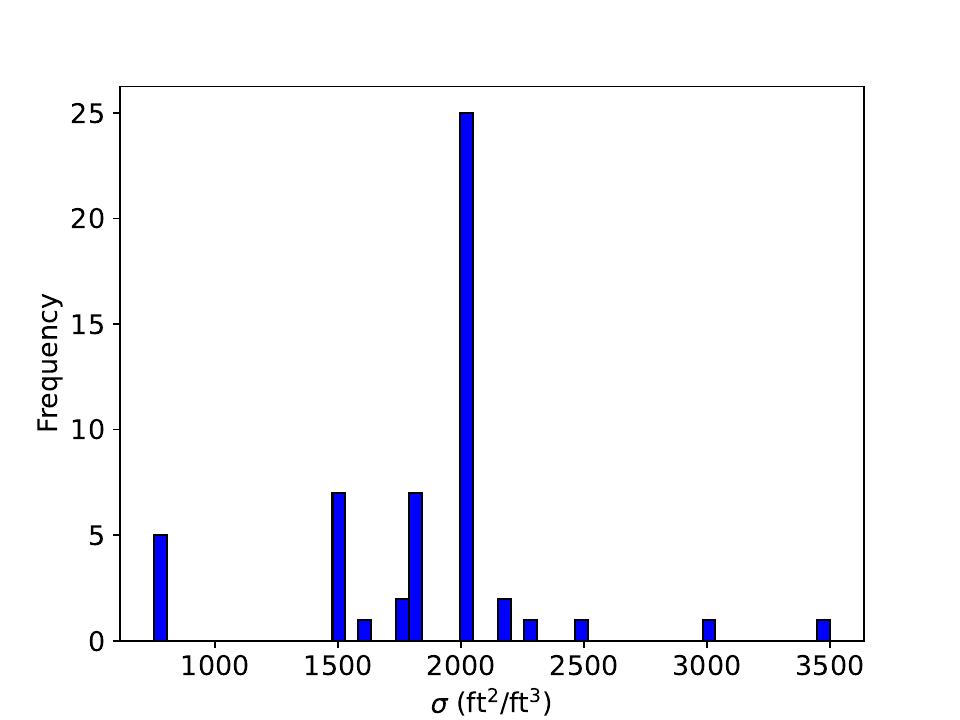}
    \caption{Surface-area-to-volume ratio.}
    \label{fig:sigma}
  \end{subfigure}\hfill
  \begin{subfigure}[t]{.45\textwidth}
    \centering
    \includegraphics[width=\linewidth]{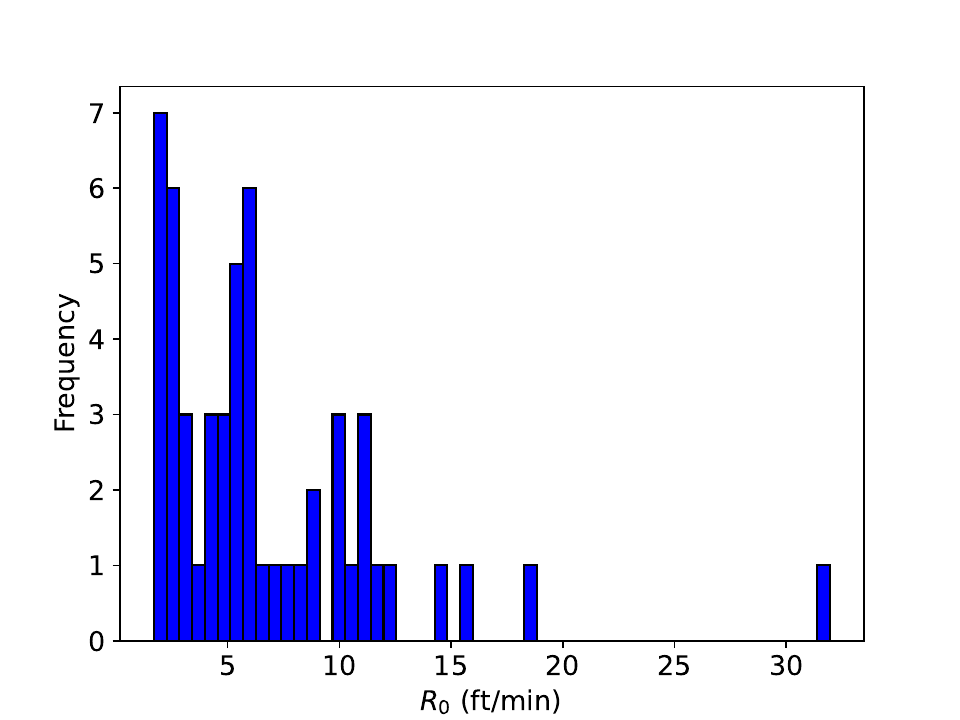}
    \caption{No-wind, No-slope rate of spread.}
    \label{fig:R0}
  \end{subfigure}

  \caption{Histograms of values considering the 53 fuel models listed by~\textcite{Andrews/2018}.}
  \label{fig:estimates}
\end{figure}

\subsection*{Slope}

As discussed in Section~\ref{sec:rothermel} and illustrated by Figure~\ref{fig:rothermel}, the slope factor $\Phi_s$ in Rothermel's model depends on the tangent of the slope angle $\phi$. More specifically, given two adjacent vertices $u, v \in V$, the tangent of the slope from $u$ to $v$ is $$A_{uv} = (v_z - u_z) / d,$$ where $d$ is the distance between $u$ and $v$ in the $xy$-plane.

While the $x$ and $y$ coordinates of each vertex are determined by the grid size, the $z$ coordinates are generated using Perlin noise~\parencite{Perlin/2002}. This gradient-based noise function is commonly used in procedural terrain generation to produce natural-looking variations in elevation. Each vertex $v \in V$ is assigned a Perlin noise value $P_v$ in the range $[0, 1]$, which we map to the interval $[0, N_z]$ feet. By adjusting the maximum elevation $N_z$ we control the steepness of the terrain. We define our slope categories based on National Wildfire Coordinating Group field guides, which classify slopes of approximately 0$^\circ$ as flat, 20$^\circ$ as moderate, and 40$^\circ$ as steep~\parencite{NWCG2021FieldGuide}. Our choice of $N_z$ values as defined in Table~\ref{tab:gfactors} produce average slope angles consistent with the slope classifications in the literature.

The slope factor $\Phi_s$ in Rothermel's model depends also on the packing ratio $\beta$. By analyzing 53 fuel models compiled by \textcite{Andrews/2018}, we found that small values of $\beta$ are prevalent across diverse vegetation types, with a significant concentration near $0.005$ (Figure~\ref{fig:beta}). Consequently we adopt $\beta = 0.005$ as a representative value when computing $\Phi_s(A_{uv})$.

Finally, we impose a technical cap on terrain steepness. Because the empirical coefficients for $\Phi_s$ derived by \textcite{Rothermel/1972} were developed for angles below 40$^\circ$, and slopes exceeding 45$^\circ$ are not addressed in the literature, we cap the slope angle at 45$^\circ$ during the generation process to maintain the physical validity of the model.

\subsection*{Wind}

The wildfire suppression problem assumes static fire travel times, meaning wind effects remain constant over the time horizon. To incorporate wind influence, we define a time-invariant wind field. Introducing time-varying wind would require local wind forecasts throughout that horizon and would transfer their uncertainty into a different dynamic or stochastic optimization problem. This field is characterized by a predominant wind direction, represented by a unit vector $\textbf{w} \in \mathbb{R}^2$. To introduce local variations, we assign a wind velocity vector $\textbf{w}_{uv}$ to each pair of adjacent vertices $u, v \in V$. This vector accounts for wind conditions around the corresponding grid cells and is derived by perturbing the primary wind direction $\textbf{w}$ in both magnitude and orientation. Figure~\ref{fig:instances:wind} illustrates a wind field for the small grid introduced in Figure~\ref{fig:instances:landscape}.

\begin{figure}[]
  \centering
  \includegraphics[scale=0.4]{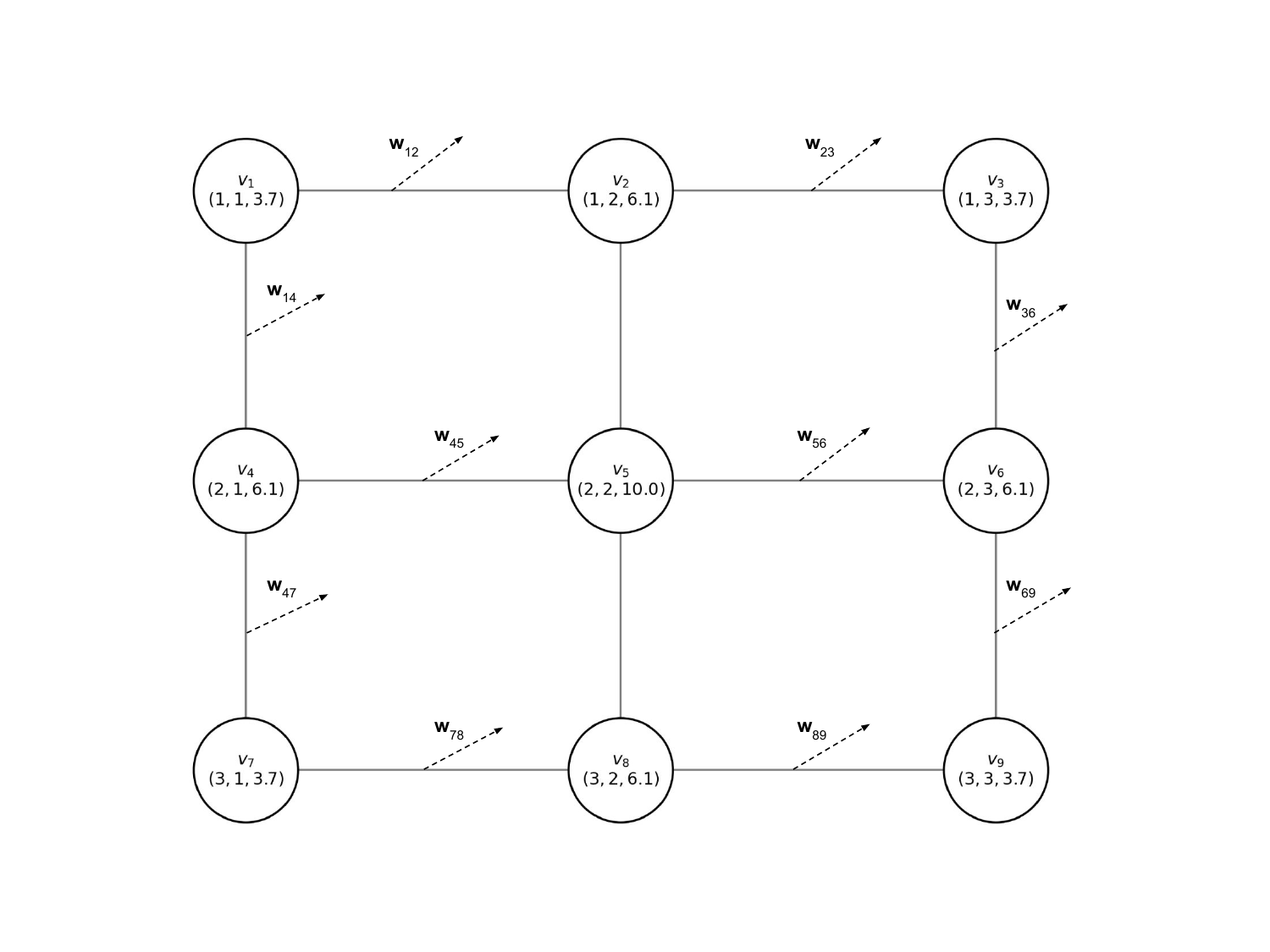}
  \caption[Illustration of the wind field used in the instance model.]{Illustration of the wind field used in the instance model. Each vector $\textbf{w}_{uv}$ represents the wind velocity between adjacent vertices $u$ and $v$, incorporating both directional and magnitude perturbations relative to the predominant wind direction. These variations are introduced using Perlin noise.}
  \label{fig:instances:wind}
\end{figure}

To introduce variability, we apply two Perlin noise functions to the wind vector $\textbf{w}$. The first noise, mapped to the interval $[-\pi / 6, \pi / 6]$, introduces directional deviations from the main wind direction. The second noise determines wind speed and is mapped to an interval based on the wind category.

To categorize the magnitude of the local wind vectors $\textbf{w}_{uv}$ we employ the Modified Beaufort Scale \parencite{Britannica/2025}, which provides standard ranges for 20-foot wind speeds (e.g., ``Light'', ``Moderate'', ``Strong''). Since Rothermel's model requires the midflame wind speed as input, we convert the 20-foot wind speed ranges defined by the Beaufort scale using a wind adjustment factor of 0.3. This specific value is selected as it represents a common value used to estimate midflame wind speed under many typical field conditions \parencite{Andrews/2012}. Table~\ref{tab:gfactors} presents the selected wind categories and the corresponding adjusted midflame wind speed ranges.

For any pair of adjacent vertices, let $\pi_{uv}$ represent the rotational perturbation and $s_{uv}$ the wind speed scaling factor. The resulting wind velocity vector is $$\textbf{w}_{uv} = s_{uv} R(\pi_{uv}) \textbf{w},$$ where $R(\pi_{uv})$ is the rotation matrix that rotates the vector $\textbf{w}$ by an angle $\pi_{uv}$. Since wind conditions are assumed to be symmetric between adjacent vertices, we enforce $\textbf{w}_{uv} = \textbf{w}_{vu}$, which implies $s_{uv} = s_{vu}$ and $\pi_{uv} = \pi_{vu}$.

For each pair of adjacent vertices $u, v \in V$, let $\textbf{n}_{uv} \in \mathbb{R}^2$ denote the unit direction vector from $u$ to $v$ in the $xy$-plane. The wind speed component along this direction is given by $$U_{uv} = \textbf{w}_{uv}^{T} \textbf{n}_{uv}.$$

As with the slope factor, the wind factor $\Phi_w$ in Rothermel's model depends not only on wind speed $U_{uv}$, but also on the surface-area-to-volume ratio $\sigma$ and the relative packing ratio $\beta_{rel}$. To choose appropriate values for these parameters, we analyzed the 53 fuel models compiled by \textcite{Andrews/2018}. Figures~\ref{fig:sigma} and~\ref{fig:beta_rel} show the distribution of $\sigma$ and $\beta_{rel}$ across the dataset. The histogram for $\sigma$ shows a clear mode at 2000~ft$^2$/ft$^3$, and the histogram for $\beta_{rel}$ indicates that most values are concentrated close to 1. Consequently, we set $\sigma = 2000$~ft$^2$/ft$^3$ and $\beta_{rel} = 1$ in all our computations of the wind factor $\Phi_w(U_{uv})$.

\subsection*{Fire Propagation}

Rothermel's model defines a rate of spread assuming that fuel bed and the environmental parameters remain static. Consequently, we assume that the fuel composition inside each cell is homogeneous, meaning that the no-wind, no-slope rate of spread is constant within the cell. For each vertex $v \in V$, then, we define $R_0(v)$ as the base fire spread rate. Similar to our approach for wind and slope, we analyzed the 53 fuel models compiled by \textcite{Andrews/2018} (Figure~\ref{fig:R0}) and observed that most values are concentrated between $1$ ft/min and $15$ ft/min. We generate values of $R_0(v)$ using Perlin noise, mapping the noise to the range $[1, 15]$ ft/min.

The rate of spread of a fire depends on the directional influence of wind and slope (Albini's extensions). In our model, this is determined by examining the signs of wind speed $U_{uv}$ and slope tangent $A_{uv}$. Specifically, $U_{uv} \geq 0$ indicates a headfire, and $U_{uv} < 0$ a backfire. Likewise, $A_{uv} \geq 0$ corresponds to upslope spread, and $A_{uv} < 0$ corresponds to downslope spread. The resulting rate of spread $R(u; \textbf{n}_{uv})$ within cell $u$ in the direction of $\textbf{n}_{uv}$ is given by $R(u; \textbf{n}_{uv}) = R_0(u) r(u; \textbf{n}_{uv})$, where the multiplier $r$ is defined by
$$
r(u; \textbf{n}_{uv}) =
\begin{cases}
  1 + \Phi_w(U_{uv}) + \Phi_s(A_{uv}),               & \text{if $A_{uv} \geq 0, U_{uv} \geq 0$}, \\
  1 + \max \{0, \Phi_w(U_{uv}) - \Phi_s(A_{uv})\},   & \text{if $A_{uv} < 0, U_{uv} \geq 0$},     \\
  1 + \max \{0, \Phi_s(A_{uv}) - \Phi_w(|U_{uv}|)\}, & \text{if $A_{uv} \geq 0, U_{uv} < 0$},    \\
  1,                                                 & \text{if $A_{uv} < 0, U_{uv} < 0$}.
\end{cases}
$$

To determine the fire propagation time from the center of cell $u$ to the center of cell $v$, we account for the transition between the different spread velocities of the two cells. We compute this by taking the distance $d(u,v)$ over the harmonic mean of the velocities in the given direction. For any arc $uv \in A$, the fire propagation time is $$t_{uv} = d(u,v) \frac{R(u; \textbf{n}_{uv}) + R(v; \textbf{n}_{uv})}{2 R(u; \textbf{n}_{uv}) R(v; \textbf{n}_{uv})}.$$

\subsection*{Optimization Horizon}

As noted by~\textcite{Harris/2022}, in the $20\times20$ grid instances generated by~\textcite{Mendes/2022}, only a small subset of the 400 vertices burns before the horizon $H$. The instances proposed by~\textcite{Harris/2022} suffer from this issue to a lesser extent. Vertices that do not burn within $H$ can be removed from the instance since they never contribute to the objective function, which can make instances simpler than their nominal size suggests.

Simply setting the horizon to exceed the latest fire arrival time for any vertex under a free-burning process could result in an excessively large $H$ due to a few outliers that burn very late. Conversely, capping the horizon at a reasonable value, such as 48 hours, might inadvertently remove too many vertices, recreating the problem found in earlier literature.

To address this, consider $B^{\A_0}_t$, the set of vertices burned at time $t$ under a free burning process. We define $q(p)$ as the time when $p\%$ of the total vertices have burned:
$$q(p) = \max \left\{t : |B^{\A_0}_t| / |V| \leq p\% \right\}.$$

In this formulation, $q(p)$ is measured in hours. To ensure that the optimization horizon is robust, we define
$$H = \max\!\left\{\, \min\!\big\{\max \{q(100),\, 24\text{h}\},\, 48\text{h}\big\},\; q(70) \right\}.$$
This logic ensures that horizon $H$ is at least 24 hours and that at least $70\%$ of the vertices burn before the horizon, and hence cannot be trivially removed from the instance. It also ensures that if the final $30\%$ of vertices burn very late, the horizon is capped at 48 hours.

\subsection*{Fire Suppression Resources}

To complete the WSP instance definition, we must specify the number of available resources $k$, the sequence of release times $t_1 \leq t_2 \leq \cdots\leq t_{T}$, and the subset $R_t \subseteq R$  of resources released at time $t \in [0,H]$. Additionally, we must define the suppression delay $\Delta$ provided by each resource.

To maintain consistency across different instances, we define the delay $\Delta$ relative to the optimization horizon $H$. As shown in Table~\ref{tab:gfactors} we consider three possible delay categories. Under high delay ($\Delta = H$), a resource effectively blocks fire spread, since any fire arrival time passing through a protected vertex will be at least $H$. For low and medium delay values, resources do not fully block fire spread but still significantly slow it down.

The number of available resources $k$ is scaled proportionally to the  grid size $n$. Table~\ref{tab:gfactors} shows the number of available resources per grid size. An instance with a grid size $20 \times 20$ and many resources, for example, has $40$ available resources. This scaling reflects the observation that, in planar graphs optimal resource allocations often form fire suppression lines. In instances with a moderate amount of resources, then, there are enough resources to construct a firebreak that spans the dimension of the landscape. In contrast, an instance with few resources may require a different allocation strategy.

We now turn our attention to the sequence of times at which resources become available. First, we establish the number of decision points $T$ where resources are released. This parameter is important because it directly dictates the size of the MIP model and the complexity of the master problem of the Benders decomposition proposed by~\textcite{Harris/2022}. It also influences the depth of the search tree of the beam search of~\textcite{Delazeri/2024a}.

For a given number of decision points $T$, the total number of available resources $k$ is distributed evenly across the time instants by first setting $r_i = \floor{k / T} + [i \leq k \mod T]$ for $i \in [T]$, and then randomly permuting quantities $r_i$, to avoid biasing instances towards higher suppression capacity in earlier time instants.

Lastly, we define the release times $t_1 \leq t_2 \cdots \leq t_T$ relative to the progression of fire, specifically the percentage of the landscape burned by a given time under a free-burning process. This allows a clearer understanding of the wildfire's state when suppression resources arrive. By anchoring the first release time to a burn percentage, we avoid two extremes: releasing resources too early, when the fire perimeter is small enough to be easily contained, or too late, when most of the landscape has already burned and the remaining problem becomes trivial.

Recall that $q(p)$ is the time when $p\%$ of the vertices burned under a free-burning process. Decision points are equally spaced over the interval between the first and the last release times, both of which are determined via $q(p)$ for specific values of $p$. Table~\ref{tab:gfactors} summarizes the possible configurations.

\section{Experiments} \label{sec:experiments}

In this section, we present computational experiments to analyze the performance of the MIP model described in Section~\ref{sec:mip_formulation}, and to benchmark all algorithms available in the literature using the proposed instance generator. All the experiments were done on a platform equipped with a 3.5 GHz AMD Ryzen 9 3900X 12-Core processor, 32 GB of main memory, and Ubuntu Linux 20.04.

The experiments compare the beam search of \textcite{Delazeri/2024a} (IBS), the logic-based Benders Decomposition of \textcite{Harris/2022} (LBBD), the iterated local search of \textcite{Mendes/2022}, the MIP formulation of \textcite{Alvelos/2018} with the improvements proposed by \textcite{Harris/2022} ($\text{MIP}_{A}$), and the MIP formulation described in Section~\ref{sec:mip_formulation} (MIP). As a baseline, we also consider a random search algorithm (RS), which is described in Appendix~\ref{app:new_instances:random_search}.

We used the C++ implementation of IBS available in the literature, compiled with GCC 13.3.1 using maximum optimization. The random search was also implemented in C++ and compiled in the same way. The implementation of LBBD and ILS are in Python and were provided by \textcite{Harris/2022}.

For the exact algorithms (LBBD, $\text{MIP}_{A}$, and MIP) we employed Gurobi 10.0.3 as the underlying solver with default settings. Regarding the heuristics IBS and ILS, we adopted parameter settings recommended by \textcite{Mendes/2022} and \textcite{Delazeri/2024a}, with two adjustments: the beam search parameter $\hat{p}$ was set to $1$, and the number of random initial solutions for ILS was also set to $1$. We observed that these settings lead to better results.

All instances, detailed experimental results, and the instance generator are available at \url{https://github.com/gutodelazeri/WSP}. The codes of LBBD, ILS, and IBS are from the respective papers.

\subsection{Literature Instances} \label{sec:experiments:literature}

We consider the set of instances proposed by \textcite{Harris/2022}, consisting of 16 problems defined on $20 \times 20$ grid graphs. In all these instances, the fire starts at the center of the grid, and the optimization horizon is fixed at $70$. The instances are divided into two groups, LA and LB, each containing $8$ instances. Within each group, all instances share identical resources release schedules and delay magnitudes. The instance set is summarized in Table~\ref{tab:algorithms:experiments:instances}. The optimal solution is known for every instance.

\begin{table}[]
  \centering
  \setlength{\tabcolsep}{10pt}
  \caption{Summary of literature instances.}\label{tab:algorithms:experiments:instances}
  \begin{tabular}{ccccccccc}
    \toprule
        & \multicolumn{6}{c}{Resources per time instant} &    &       \\ \cmidrule(lr{.75em}){2-7}
  Group  & 10     & 20     & 30     & 40     & 50     & 60    & $H$  & $\Delta$ \\
  \midrule
  LA     & 3      & 3      & 3      & 3      & 0      & 0     & 70 & 50    \\
  LB     & 3      & 3      & 3      & 3      & 3      & 3     & 70 & 30    \\
  \bottomrule
  \end{tabular}
\end{table}

Fire travel times are modeled using direction-dependent uniform distributions to capture the influence of wind. Directions aligned with the wind are assigned shorter propagation times, while opposing directions are slower. For each arc, the travel time is sampled from a distribution determined by the arc's direction. For example, to model fire propagation times under the influence of a south-east wind, arcs pointing south use $U(2, 4)$ and arcs pointing east use $U(4, 6)$, while arcs pointing north and west use $U(7,9)$ and $U(6,8)$. See Table~5 in \textcite{Harris/2022} for details.

\begin{table}
    \newcolumntype{C}{ @{}>{${}}c<{{}$}  @{}}
    \centering
    \caption[Comparison of absolute deviation and MIP gap between MIP models]{Comparison of the absolute deviation from the optimal objective value and the MIP gap (\%) between our MIP model (MIP), the model of \textcite{Alvelos/2018} $\text{MIP}_{A}$, and the logic-based Benders Decomposition of \textcite{Harris/2022}. Values are mean $\pm$ standard deviation over $10$ runs. In cases when $\text{MIP}_{A}$ failed to find a feasible solution, the worst outcome (all vertices burn) was assumed when computing the statistics. The final two rows show summary statistics for each instance group. The best average absolute deviation for each instance is highlighted in bold.} \label{tab:literature_instances}

    \smallskip
    \begin{tabular}{l *6{rCl}}
    \toprule
     & \multicolumn{9}{c}{Avg. Abs. Dev. to BKV} & \multicolumn{9}{c}{Avg. MIP gap (\%)} \\ \cmidrule(lr){2-10} \cmidrule(lr){11-19}
          & \multicolumn{3}{c}{MIP} & \multicolumn{3}{c}{$\text{MIP}_{A}$} & \multicolumn{3}{c}{LBBD} & \multicolumn{3}{c}{MIP} & \multicolumn{3}{c}{$\text{MIP}_{A}$} & \multicolumn{3}{c}{LBBD} \\
    \cmidrule(lr){2-4} \cmidrule(lr){5-7} \cmidrule(lr){8-10}  \cmidrule(lr){11-13} \cmidrule(lr){14-16} \cmidrule(lr){17-19}
    LA0 & $0.2$ & \pm & $0.4$ & $44.1$ & \pm & $24.0$  & $\textbf{0.0}$ & \pm & $0.0$ & $16.7$ & \pm & $1.4$ & $126.6$ & \pm & $20.0$ & $0.0 $ & \pm & $0.0$ \\
    LA1 & $\textbf{0.0}$ & \pm & $0.0$ & $78.1$ & \pm & $24.5$  & $\textbf{0.0}$ & \pm & $0.0$ & $ 0.0$ & \pm & $0.0$ & $114.9$ & \pm & $17.6$ & $0.0 $ & \pm & $0.0$ \\
    LA2 & $1.2$ & \pm & $0.9$ & $52.6$ & \pm & $26.8$  & $\textbf{0.0}$ & \pm & $0.0$ & $15.0$ & \pm & $1.1$ & $135.1$ & \pm & $28.1$ & $0.0 $ & \pm & $0.0$ \\
    LA3 & $0.2$ & \pm & $0.4$ & $46.0$ & \pm & $24.9$  & $\textbf{0.0}$ & \pm & $0.0$ & $ 7.0$ & \pm & $0.7$ & $70.1 $ & \pm & $16.9$ & $0.0 $ & \pm & $0.0$ \\
    LA4 & $\textbf{0.0}$ & \pm & $0.0$ & $94.0$ & \pm & $12.2$  & $\textbf{0.0}$ & \pm & $0.0$ & $ 0.4$ & \pm & $0.8$ & $123.6$ & \pm & $5.8 $ & $0.0 $ & \pm & $0.0$ \\
    LA5 & $\textbf{0.0}$ & \pm & $0.0$ & $60.4$ & \pm & $14.8$  & $\textbf{0.0}$ & \pm & $0.0$ & $ 1.0$ & \pm & $1.3$ & $78.0 $ & \pm & $11.0$ & $0.0 $ & \pm & $0.0$ \\
    LA6 & $\textbf{0.0}$ & \pm & $0.0$ & $63.7$ & \pm & $23.4$  & $\textbf{0.0}$ & \pm & $0.0$ & $12.2$ & \pm & $1.3$ & $97.3 $ & \pm & $13.9$ & $0.0 $ & \pm & $0.0$ \\
    LA7 & $\textbf{0.0}$ & \pm & $0.0$ & $48.1$ & \pm & $18.1$  & $\textbf{0.0}$ & \pm & $0.0$ & $ 4.1$ & \pm & $0.6$ & $64.0 $ & \pm & $11.4$ & $0.0 $ & \pm & $0.0$ \\
    LB0 & $\textbf{3.8}$ & \pm & $3.0$ & $75.1$ & \pm & $32.1$  & $6.6$ & \pm & $3.2$ & $32.9$ & \pm & $2.4$ & $185.7$ & \pm & $29.6$ & $12.7$ & \pm & $3.6$ \\
    LB1 & $2.4$ & \pm & $1.6$ & $58.6$ & \pm & $34.7$  & $\textbf{0.0}$ & \pm & $0.0$ & $17.8$ & \pm & $1.5$ & $121.1$ & \pm & $24.2$ & $0.0 $ & \pm & $0.0$ \\
    LB2 & $2.8$ & \pm & $1.7$ & $51.1$ & \pm & $30.8$  & $\textbf{2.3}$ & \pm & $2.4$ & $27.6$ & \pm & $0.9$ & $150.0$ & \pm & $30.5$ & $8.6 $ & \pm & $2.8$ \\
    LB3 & $1.9$ & \pm & $4.3$ & $47.9$ & \pm & $29.7$  & $\textbf{1.1}$ & \pm & $2.3$ & $15.5$ & \pm & $2.6$ & $85.4 $ & \pm & $19.7$ & $5.5 $ & \pm & $3.7$ \\
    LB4 & $\textbf{6.4}$ & \pm & $3.0$ & $80.8$ & \pm & $23.5$  & $7.6$ & \pm & $5.1$ & $23.4$ & \pm & $1.7$ & $146.9$ & \pm & $18.4$ & $9.9 $ & \pm & $3.6$ \\
    LB5 & $3.9$ & \pm & $2.9$ & $46.0$ & \pm & $23.1$  & $\textbf{3.8}$ & \pm & $3.6$ & $19.0$ & \pm & $2.0$ & $86.7 $ & \pm & $13.4$ & $8.3 $ & \pm & $2.3$ \\
    LB6 & $\textbf{4.4}$ & \pm & $3.8$ & $51.2$ & \pm & $27.6$  & $5.2$ & \pm & $4.5$ & $22.0$ & \pm & $2.6$ & $101.4$ & \pm & $16.9$ & $8.3 $ & \pm & $3.0$ \\
    LB7 & $\textbf{1.5}$ & \pm & $3.1$ & $56.6$ & \pm & $24.4$  & $4.5$ & \pm & $2.3$ & $15.9$ & \pm & $2.1$ & $77.3 $ & \pm & $13.1$ & $10.1$ & \pm & $2.2$ \\
    \midrule
    LA  & $0.2$ & \pm & $0.5$ & $60.9$ & \pm & $26.4$  & $\textbf{0.0}$ & \pm & $0.0$ & $ 7.0$ & \pm & $6.5$ & $101.2$ & \pm & $30.6$ & $0.0$ & \pm & $0.0$  \\
    LB  & $\textbf{3.4}$ & \pm & $3.3$ & $58.4$ & \pm & $29.8$  & $3.9$ & \pm & $4.0$ & $21.7$ & \pm & $6.1$ & $119.3$ & \pm & $41.8$ & $7.9$ & \pm & $4.5$ \\
    \bottomrule
    \end{tabular}
\end{table}

The three exact algorithms available in the literature (MIP, $\text{MIP}_{A}$, and LBBD) were executed 10 times per instance with a time limit of 600 seconds per run. We evaluate performance using the absolute deviation from the optimal objective value and the MIP gap (\%). For any instance $I$ and algorithm $A$, let $A(I)$ be the objective value obtained, $A^{\star}(I)$ the known optimal value, and $L(I)$ the best lower bound found within the time limit. The absolute deviation is calculated as $A(I) - A^{\star}(I)$ and the MIP gap as $(A(I) - L(I)) / L(I)$. We use the absolute deviation as our performance metric because the 16 instances are structurally similar (they have nearly the same number of vertices) and their optimal objective values are of the same order of magnitude.

The results are shown in Table~\ref{tab:literature_instances}. The first group of columns show that our MIP formulation slightly outperforms and significantly surpasses $\text{MIP}_{A}$, in terms of solution quality. MIP obtained an average absolute deviation of $0.2$ on the LA group and of $3.4$ on the LB group, compared to $0$ and $3.9$ for LBBD, and $60.9$ and $58.4$ for $\text{MIP}_{A}$.

Regarding the MIP gap, as shown by the second group of columns, both MIP models obtain high gaps, and fail to prove optimality in any of the replications. In contrast, LBBD proved optimality for all instances in the LA group across all replications, and reached an average MIP gap of $7.9\%$ for the LB group.

In summary, our results show that the MIP formulation described in Section~\ref{sec:mip_formulation} is highly competitive with LBBD regarding solution quality on existing benchmarks. However, large MIP gaps indicate that there is room for improvement, specifically by strengthening the relaxation to assist the solver in proving optimality.

\subsection{Instance Generator} \label{sec:experiments:new_instances}

The instance generator defines eight experimental factors, as detailed in Table~\ref{tab:gfactors}. To evaluate their influence, we study these factors in pairs, while keeping the others at their default values. We have organized these eight factors into four groups (Table~\ref{tab:gfactors}) based on the premise that factors in the same group are likely to have similar effects on algorithm behavior. For instance, grid size and the number of decision points both govern the size of the models solved by MIP and LBBD, whereas environmental factors such as wind and slope dictate the velocity and direction of fire propagation.

We select $30 \times 30$ grids as the default size because, on larger grids, most algorithms fail to make progress within reasonable time limits, which would obscure the influence of other experimental factors. The default number of decision points is set to 10, following the work of \textcite{Avci/2024}, who modeled real-world suppression efforts with around ten decisions over a 10-hour horizon. Since our test instances typically span 24 to 48 hours, using 10 to 20 decision points is realistic.

We use a ``high'' suppression delay as the default, implying that suppression resources block fire spread at the protected vertices. The default number of resources is set to ``moderate'' ($k=n$ for an $n\times n$ grid) to ensure to reveal the impact of other factors; at lower resource levels, the influence of other factors is weak.

The slope and wind are set to moderate values to avoid extreme fire spread velocities in the baseline case. The effect of extreme velocities is analyzed separately. Finally, the first release time is set to ``early'' (minimal burned area) and the final release time is set to ``very late'' (extending to the time horizon's end).

The algorithm parameters are the same as those used in Section~\ref{sec:experiments:literature}, with the exception of time limits. Since the new instances have larger grids, we scale the time limit with grid size at a rate of 1.5 seconds per grid cell. This results in limits of 600 seconds for $20 \times 20$ grids, 1350 seconds for $30 \times 30$, 2400 seconds for $40 \times 40$, and 9600 seconds for $80 \times 80$ grids.

For each factor pair, we generate five unique instances by varying the instance seed. Each algorithm is run ten times per instance. To generate the instances, we set $N_{xy}$ to 26240 feet, leaving us with square landscapes of approximately 64 km$^2$. The experiments were conducted using the same hardware and software environment detailed in Section~\ref{sec:experiments:literature}.

Our analysis employs the Skillings-Mack test, a non-parametric alternative to the repeated measures ANOVA~\parencite{Skillings/1981}. Skillings-Mack (SM) considers a design with $n$ blocks and $k$ treatments, and each cell (block-treatment combination) may have zero or more replications. Our design remains balanced, with each cell holding $c=10$ replications.  Within each block $i$, we rank all $ck$ replications (lower ranks are better), assigning average ranks to tied values. For treatment $j$, we define the block-wise mean rank of its $c$ replication ranks $$R_{ij} = \frac{1}{c}\sum_{r\in[c]} \text{rank}(z_{ijr}),$$ where $z_{ijr}$ is the response of the $r$-th replication for treatment $j$ in block $i$. Mean ranks $R_{ij}$ serve as the inputs for calculating the treatment scores in both the omnibus test and the post-hoc tests. For post-hoc comparisons we assign to each treatment $j$ a score $S_j = \sum_{i\in[n]} R_{ij}$, which is the sum of the within-block ranks across all blocks. Two treatments $j$ and $j^{\prime}$ are considered significantly different if their absolute score difference exceeds a threshold $\delta$: $$\lvert S_j - S_{j^{\prime}} \rvert > \delta.$$ The value of $\delta$ is computed based on a family-wise error rate $\alpha$ for $k$ treatments.

\begin{table}[]
  \centering
  \small
  \caption{Skillings--Mack post-hoc orderings by parameter group ($\alpha=0.001$). We use $A \prec B$ to indicate that algorithm $A$ is significantly better than $B$, and $A \approx B$ when they are not significantly different, but $A$ typically performs better than $B$.}
  \label{tab:sm_orders}
  \begin{tabular}{ll}
  \toprule
  Group & Ordering \\
  \midrule
  Instance size         & IBS $\prec$ MIP $\prec$ ILS $\prec$ (LBBD $\approx$ RS) \\
  Suppression capacity  & (MIP $\approx$ IBS $\approx$ ILS) $\prec$ (RS $\approx$ LBBD) \\
  Environmental factors & IBS $\prec$ (ILS $\approx$ MIP) $\prec$ (LBBD $\approx$ RS) \\
  Release window        & IBS $\prec$ MIP $\prec$ ILS $\prec$ LBBD $\prec$ RS \\
  \bottomrule
  \end{tabular}
\end{table}

Our first analysis evaluates relative competitiveness of the algorithms. We apply the SM test within each parameter group, where algorithms are treatments and instances are blocks, using the objective value as the response variable at a significance level $\alpha=0.001$. All omnibus tests rejected equality, so we conducted all-pairs post-hoc comparisons. Table~\ref{tab:sm_orders} shows the results.

Across most groups, IBS ranks first, often by a statistically significant margin. MIP and ILS are indistinguishable in two of the groups, and MIP is better in the other two. In contrast, LBBD and RS are indistinguishable in most cases, except within the release window group, where LBBD is better. In summary, the results indicate that IBS is the most reliable default solver, with our MIP model being as the most effective runner-up.

\begin{figure}[]
  \centering
  \begin{subfigure}{0.48\textwidth}
    \centering
    \includegraphics[width=\linewidth]{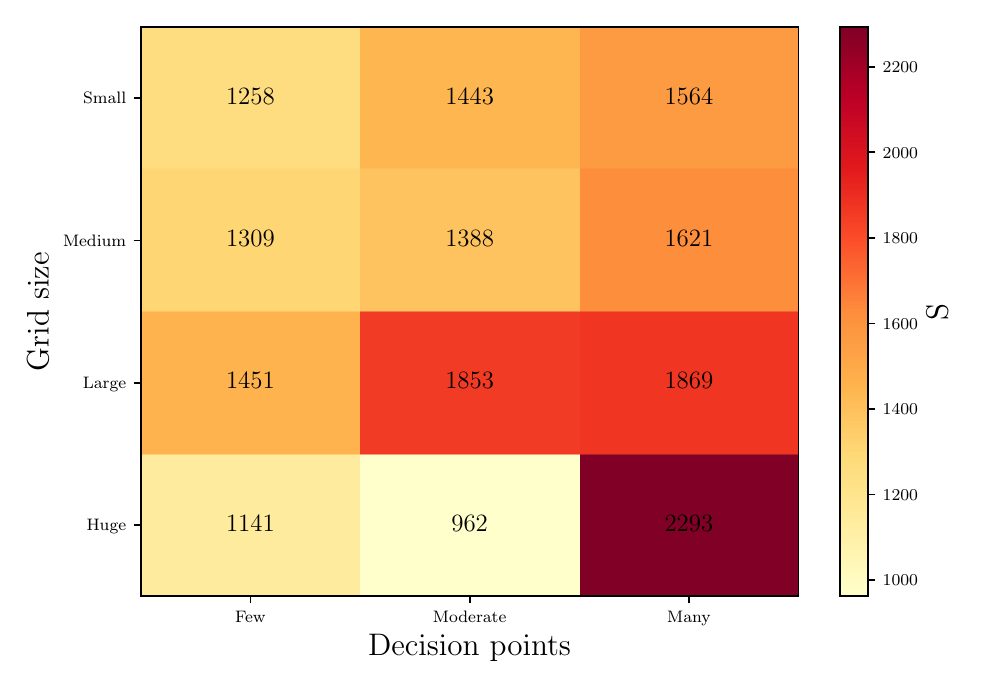}
    \caption{Instance size ($\delta=335$).}
    \label{ffig:exp:hmap:instance_size}
  \end{subfigure}
  \hfill
  \begin{subfigure}{0.48\textwidth}
    \centering
    \includegraphics[width=\linewidth]{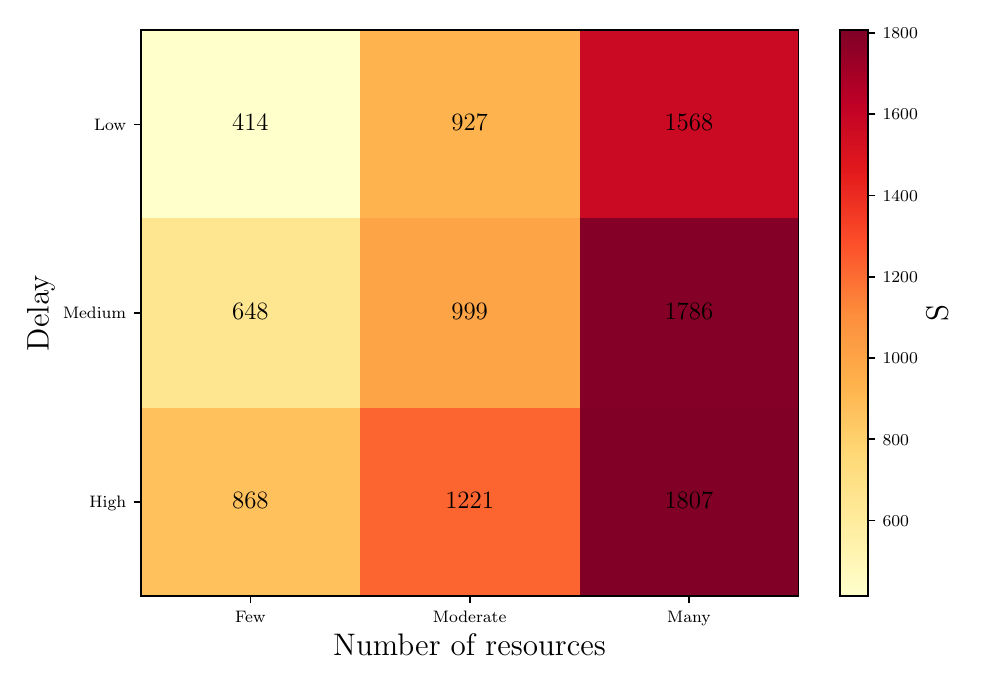}
    \caption{Suppression capacity ($\delta=244$).}
    \label{ffig:exp:hmap:suppression_capacity}
  \end{subfigure}
  \vspace{0.75em}
  \begin{subfigure}{0.48\textwidth}
    \centering
    \includegraphics[width=\linewidth]{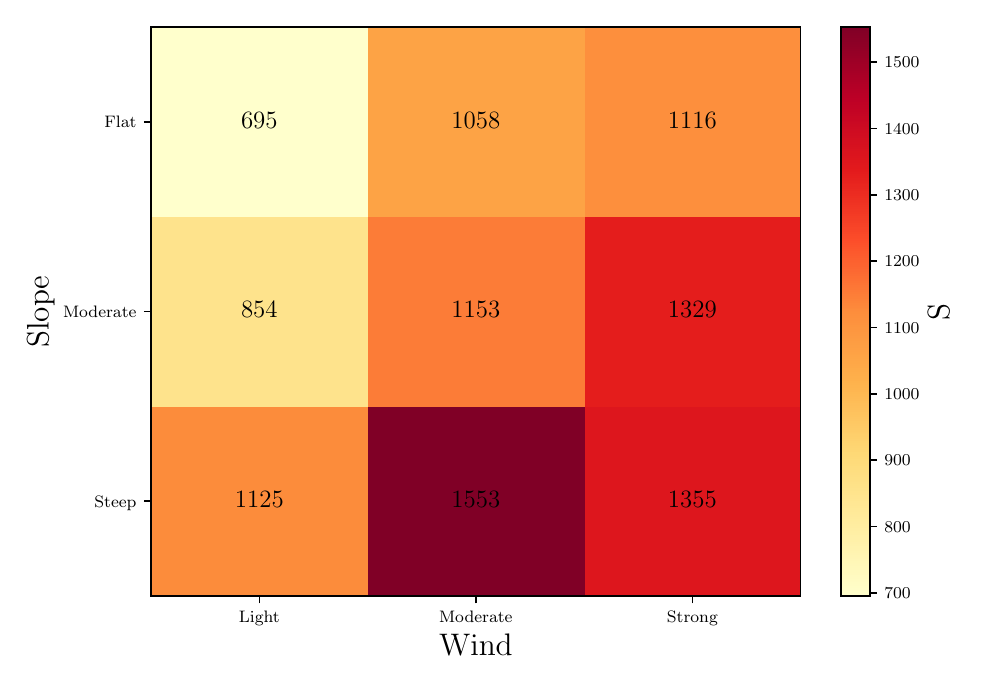}
    \caption{Environmental factors ($\delta=244$).}
    \label{ffig:exp:hmap:env_factors}
  \end{subfigure}
  \hfill
  \begin{subfigure}{0.48\textwidth}
    \centering
    \includegraphics[width=\linewidth]{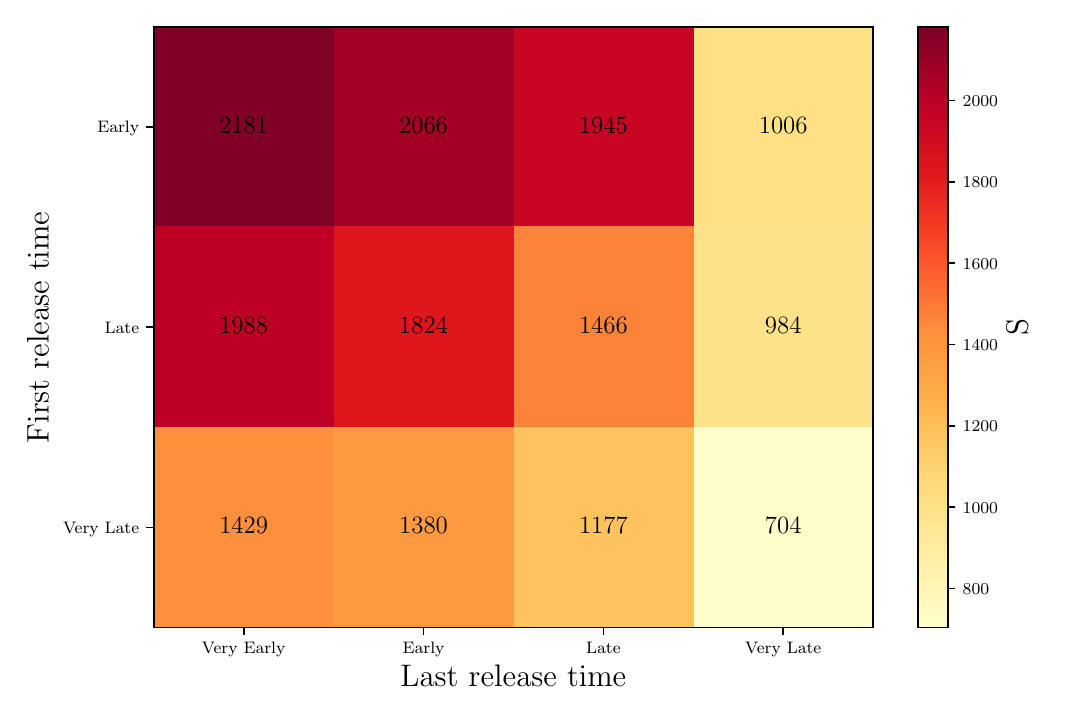}
    \caption{Release window ($\delta=335$).}
    \label{ffig:exp:hmap:release_window}
  \end{subfigure}
  \caption{Difficulty of factor combinations by group. Each cell shows the sum of within-block average ranks $S$ (lower $S$ = lower relative deviation). Within each figure, pairs with $|\Delta S|>\delta$ are significantly different at $\alpha=0.001$.}
  \label{fig:exp:hmap}
\end{figure}

Next, we analyze which combinations of factors pose the greatest challenges to the algorithms. Since the optimal solutions for these new, larger instances are unknown, we use the best-known value (BKV) of an instance (the minimum objective achieved over all algorithms and replications for a given instance). We measure difficulty using the relative deviation to the BKV, which is invariant to the scale of the objective function.

We again applied the SM test ($\alpha = 0.001$) to each parameter group, with factor combinations as treatments and algorithm-seed pairs as blocks. In all cases, the omnibus test rejected equality, prompting all-pairs post-hoc comparisons. Figure~\ref{fig:exp:hmap} shows the results for each parameter group via heatmaps. The color intensity represents the sum of within-block average ranks $S$; a lower $S$ value indicates that algorithms performed better on average on those instances. The provided $\delta$ value in each plot determines statistical significance: any pair of factor combinations with a rank difference $|\Delta S|>\delta$ is significantly different at the $\alpha=0.001$ level.

The results reveal distinct drivers of difficulty. In the instance size group, the number of decision points has the dominant effect. Increasing from ``few'' to ``many'' decisions markedly increases difficulty, whereas increasing grid size has a weaker effect. In suppression capacity, difficulty grows monotonically with the number of resources, and delay has a secondary influence (``Low delay-Few resources'' being the easiest and ``High delay-Many resources'' the hardest). For environmental factors, moving from ``light'' to ``moderate'' wind significantly increases difficulty. Furthermore, effects of terrain slope are more pronounced under ``light'' or ``moderate'' wind. Concerning release windows, later first/last releases consistently reduce difficulty (``Very Late-Very Late'' being the easiest and ``Early-Very Early'' the hardest).

\begin{figure}[]
  \centering
  \begin{subfigure}{0.48\textwidth}
    \centering
    \includegraphics[width=\linewidth]{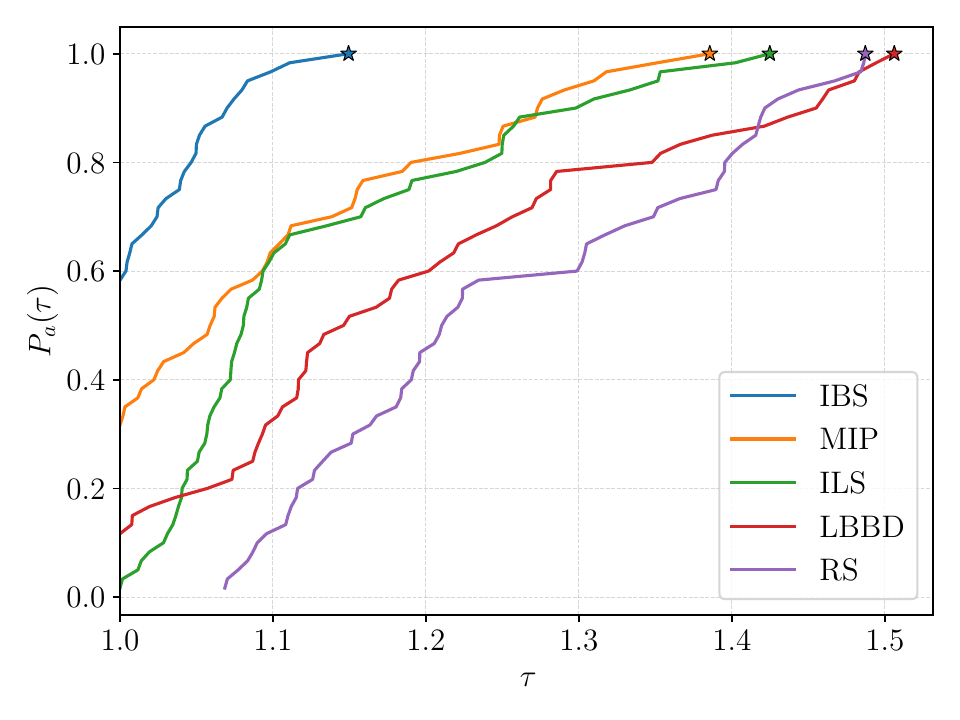}
    \caption{Instance size.}
    \label{ffig:exp:pp:instance_size}
  \end{subfigure}
  \hfill
  \begin{subfigure}{0.48\textwidth}
    \centering
    \includegraphics[width=\linewidth]{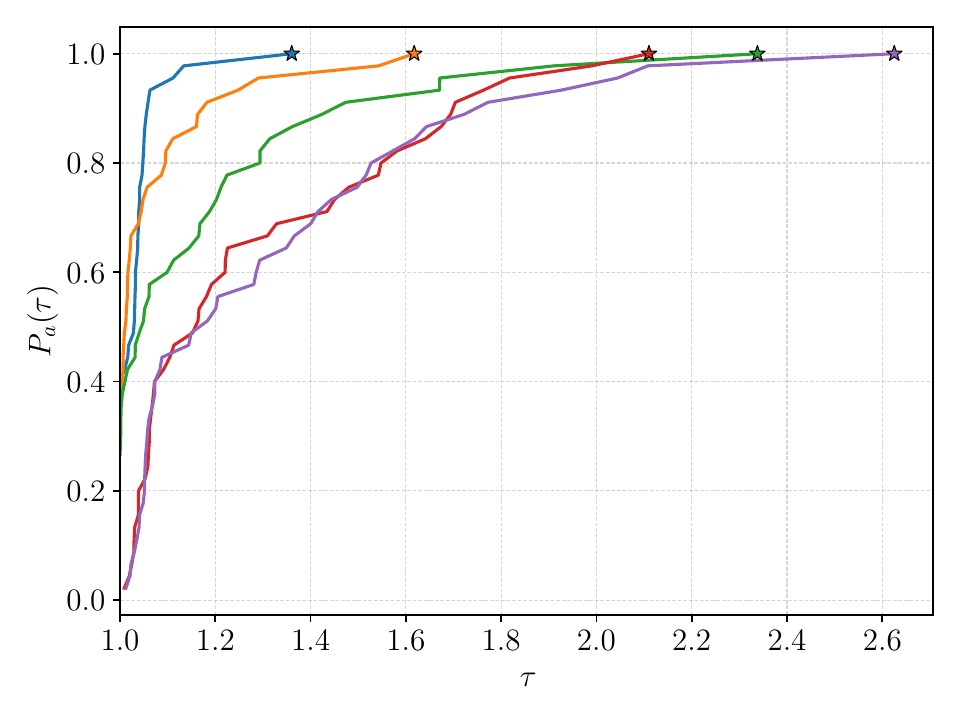}
    \caption{Suppression capacity.}
    \label{ffig:exp:pp:suppression_capacity}
  \end{subfigure}
  \vspace{0.75em}
  \begin{subfigure}{0.48\textwidth}
    \centering
    \includegraphics[width=\linewidth]{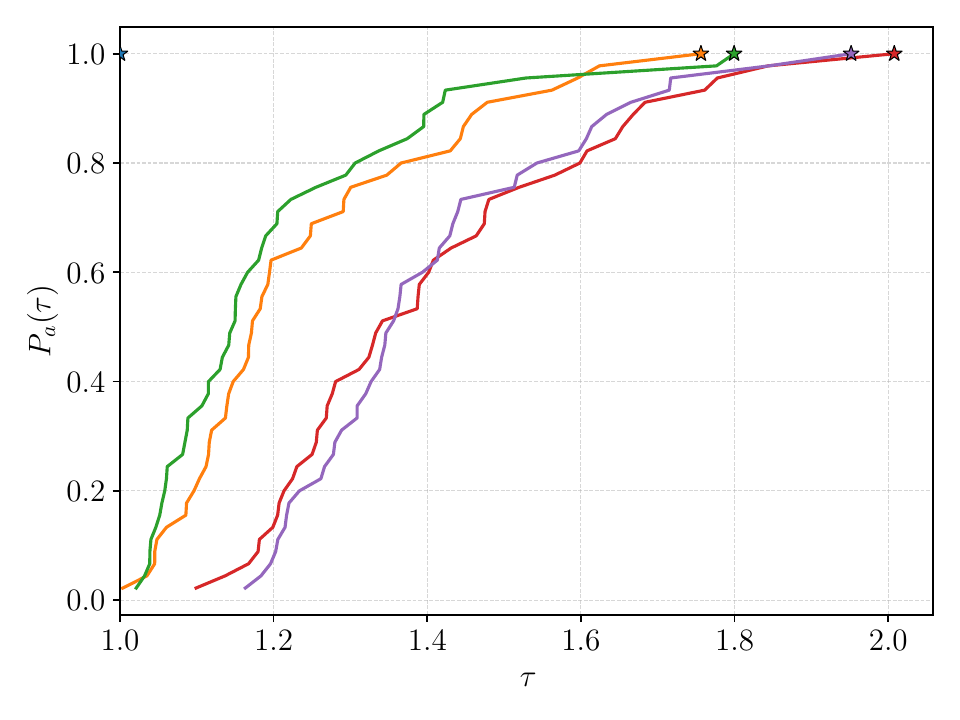}
    \caption{Environmental factors.}
    \label{ffig:exp:pp:env_factors}
  \end{subfigure}
  \hfill
  \begin{subfigure}{0.48\textwidth}
    \centering
    \includegraphics[width=\linewidth]{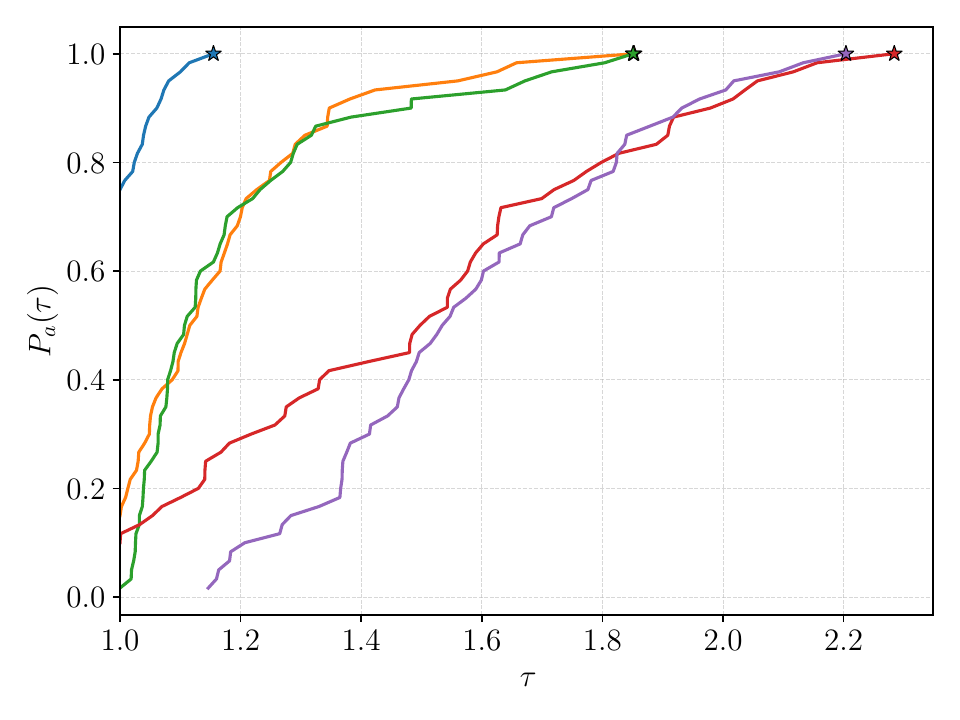}
    \caption{Release window.}
    \label{ffig:exp:pp:release_window}
  \end{subfigure}
  \caption{Performance profiles by parameter group. The performance profile $P_a(\tau)$ is the fraction of instances for which the algorithm $a$ achieves a performance ratio less than or equal to $\tau$.}
  \label{fig:exp:pp}
\end{figure}

Finally, we evaluate the relative performance of the algorithms across all generated instances using performance profiles~\parencite{Dolan/2004}. For each algorithm and instance we aggregate the objectives of the 10 runs using the median objective value $z_{a,i}$. Following \textcite{Dolan/2004}, we define $z_i^\star=\min_b z_{b,i}$ as the best median objective achieved for instance $i$ across all algorithms. We then compute the performance ratio $r_{a,i} = z_{a,i} / z_i^\star$. The performance profile $P_a(\tau)$ represents the fraction of instances for which the algorithm's performance ratio is within a factor $\tau$ of the best observed value: $$P_a(\tau)=\frac{1}{N}\lvert\{i: r_{a,i}\le\tau\}\rvert,$$ where $N$ is the number of instances (60 instances for instance size and release window, 45 instances for suppression capacity and environmental factors). For each parameter group, we plot the performance profiles for each algorithm in Figure~\ref{fig:exp:pp}.  As shown in the figures, the profiles closely align with the statistical rankings established by the Skillings-Mack tests.

\subsection{Dual Bounds} \label{sec:experiments:dual_bounds}

We finally examine the dual bounds obtained on the generated instances. Since their optimal values are unknown, we measure the quality of a dual bound $L$ relative to the best value $z^{\star}$ known for the instance by the ratio $L/z^{\star}\in[0,1]$, where one indicates a tight bound.

On the generated instances LBBD closes only $3.0\%$ of the runs ($62$ of $2100$), compared to all replications of group LA. On these instances neither method provides a bound that is useful for assessing the quality of heuristic solutions.

The relative strength of the two bounds, however, depends systematically on the suppression delay. Table~\ref{tab:dual_bounds} reports the mean ratio $L/z^{\star}$ and the percentage of runs in which the bound of the MIP is stronger. With increasing instance size both bounds deteriorate, and that of LBBD is stronger throughout. Varying the delay reverses the ranking: for a low delay the bound of the MIP is stronger in every run, and for a medium delay in $90\%$ of the runs. This is consistent with the literature instances, where LBBD proves optimality on group LA, which has the larger delay, but reaches an average gap of $7.9\%$ on group LB, which has the smaller one. A plausible explanation is that the feasibility cuts of LBBD become weak once a single resource no longer suffices to delay the fire beyond the horizon, since then combinations of resources have to be excluded, whereas the propagation constraints of our formulation do not depend on the magnitude of the delay. The effect is confined to the bounds: over the same runs our formulation obtains the better solution for a low, medium, and high delay alike, in $100\%$, $90\%$, and $88\%$ of the runs. The delay thus governs the strength of the cuts of LBBD rather than the quality of the solutions it finds.

\begin{table}[t]
  \centering
  \caption{Quality $L/z^{\star}$ of the dual bounds obtained by MIP and LBBD on the generated instances, and percentage of runs in which the bound of the MIP is stronger. In each group the stated factor is varied and all others are kept at their default values, so that the default configuration is the row ``Medium'' of the first group and the row ``High'' of the second.}
  \label{tab:dual_bounds}
  \begin{tabular}{llrrr}
    \toprule
                &        & \multicolumn{2}{c}{$L/z^{\star}$} &              \\
    \cmidrule(lr){3-4}
    Factor      & Level  & MIP  & LBBD & MIP stronger \\
    \midrule
    Grid size   & Small  & 0.64 & 0.94 & 0\%          \\
                & Medium & 0.33 & 0.74 & 0\%          \\
                & Large  & 0.26 & 0.63 & 0\%          \\
                & Huge   & 0.13 & 0.34 & 0\%          \\
    \midrule
    Delay       & Low    & 0.69 & 0.41 & 100\%        \\
                & Medium & 0.55 & 0.45 & 90\%         \\
                & High   & 0.33 & 0.74 & 0\%          \\
    \bottomrule
  \end{tabular}
\end{table}

\section{Conclusions} \label{sec:conclusions}

In this work, we have addressed the problem of allocating time-constrained fire suppression resources within a graph-based landscape. This problem has gained attention recently in the context of decision-support systems for wildfire management. Our contributions span theoretical, algorithmic, and experimental domains, as summarized below.

In Section~\ref{sec:complexity}, we established that the WSP remains strongly NP-complete when all resources are released simultaneously. We further proved strong NP-completeness for all three variants on planar graphs and for the WSP and WWSP on full weighted directed grids. The complexity of the HWSP on grids remains an interesting open question. For our specific focus, the Wildfire Suppression Problem, specialized exact and heuristic algorithms have been required because existing MIP formulations could only handle small instances, as shown by prior research and by our experiments. To resolve this, we introduced a new formulation in Section~\ref{sec:mip_formulation}. Computational experiments in Section~\ref{sec:experiments:literature} show that this formulation achieves highly competitive results when evaluated on the benchmarks from the literature.

A finding of our study is that existing instances are limited in terms of number of problems and physical realism. Furthermore, as reported in previous work, these benchmarks are no longer challenging for current methods. In Section~\ref{sec:instance_generator} we addressed this problem, by proposing a new instance generator based on Rothermel's surface fire spread model. The generator allows the creation of instances of varying degrees of difficulty and different characteristics, such as number of decision points, environmental factors, and suppression capacity. While the generator is tailored for WSP, the underlying principles used to generate the landscape and define the fire spread model can be applied to create new instances for similar problems.

In Section~\ref{sec:experiments:new_instances}, we used the generator to evaluate existing state-of-the-art algorithms and our MIP. This analysis revealed interesting effects that were obscured by limitations of literature benchmarks.

Across a wide range of factors IBS is the best algorithm, and the best choice for solving new WSP instances. Our MIP formulation and ILS are the second and third most effective, respectively. LBBD and RS are not competitive. Our experiments also identified structural drivers of problem difficulty: instances with a high number of decision points, many available resources, a high delay, or an early resource release window pose the greatest challenges to current solvers.

Our findings suggest some paths for future practical and theoretical work. While IBS is effective, it struggles on instances with few decision points. Preliminary experiments suggest that increasing the frequency with which the search tree is pruned could improve its performance. Furthermore, investigating why LBBD, previously considered a state-of-the-art approach, underperforms on realistic instances could lead to better decomposition methods. On the theoretical side, an interesting direction is the development of new combinatorial lower bounds for WSP.

\appendix

\section{Related mathematical models} \label{app:related_work}

This appendix details the connection between the models proposed by \textcite{Hof/2000} and \textcite{Wei/2011} and the MIP formulation described in Section~\ref{sec:mip_formulation}.

\textcite{Hof/2000} model the problem on a grid of cells $G = [m] \times [n]$, indexed by $c=(i,j)\in G$. The objective is to maximize the fire arrival time at target cell $t = (m,n)$, given an ignition source at cell $s = (a,b)$. For each cell $c$,  $T_c^{\text{in}}$ is the time the fire front ignites (enters) $c$, $T_c^{\text{out}}$ the time the fire front leaves $c$, $F_c$ represents the fuel remaining for combustion after suppression treatment, and $X_c$ denotes the proportion of the cell receiving treatment. Furthermore, set $\Omega_c$ contains the neighboring cells $(h,k)$ capable of igniting cell $c$, $f_c(\cdot)$ is a function relating fuel to the fire's duration within the cell, $\overline{F}_c$ is the initial fuel, $\gamma_c$ denotes the proportion of fuel removable by treatment, and $\overline{X}$ denotes the total treatment budget. The formulation presented by \textcite{Hof/2000} is as follows.
\begin{align*}
\text{max.} \quad & T_{t}^{\text{in}}  \tag{H.0} \label{hof:model:c0}                                                                                              \\
\text{s.t.} \quad & T_{s}^{\text{in}} = 0,  \tag{H.1} \label{hof:model:c1}                                                                                         \\
                  & T_{c}^{\text{in}} \leq T_{c'}^{\text{out}},                           &  & \forall c\in G, \: \forall c' \in \Omega_{c}, \tag{H.2} \label{hof:model:c2} \\
                  & T_{c}^{\text{out}} - T_{c}^{\text{in}} = f_{c}(F_{c}),                &  & \forall c\in G  \tag{H.3}, \label{hof:model:c3}                              \\
                  & F_{c} = \overline{F}_{c} - \gamma_{c} \overline{F}_{c} X_{c}, &  & \forall c\in G  \tag{H.4}, \label{hof:model:c4}                              \\
                  & 0 \leq X_{c} \leq 1,                                          &  & \forall c\in G  \tag{H.5}, \label{hof:model:c5}                              \\
                  & \sum_{c\in G} X_{c} \le \overline{X}  \tag{H.6}. \label{hof:model:c6}
\end{align*}

Constraints~\ref{hof:model:c1} sets the ignition time at the starting cell $s=(a,b)$ to $0$. Constraint~\ref{hof:model:c2} ensures a cell ignites only after the fire leaves a neighboring cell capable of igniting it. Constraint~\ref{hof:model:c3} defines the time the fire takes to pass through cell $c$ as a function of the available fuel $F_{c}$. Constraint~\ref{hof:model:c4} calculates the available fuel after treatment $X_{c}$ reduces the initial fuel $\overline{F}_{c}$. Constraint~\ref{hof:model:c5} limits the treatment proportion to the interval $[0,1]$, while constraint~\ref{hof:model:c6} restricts the total treatment budget to $\overline{X}$. The objective function maximizes the time the fire reaches the target cell $t = (m,n)$.

Note that we can merge the equality constraints~\ref{hof:model:c3} and~\ref{hof:model:c4} into constraint~\ref{hof:model:c2}, leaving us with
\begin{align*}
  T_{c}^{\text{in}} \leq T_{c'}^{\text{in}} + f_{c'}(\overline{F}_{c'} - \gamma_{c'} \overline{F}_{c'} X_{c'}), && \forall c\in G, \: \forall c' \in \Omega_{c}. \tag{H.7} \label{hof:model:c2_reformulated}
\end{align*}

In practice, $f_{c}(\cdot)$ is a linear function of the decision variable $X_{c}$, parameterized by the amount of available fuel $\overline{F}_{c}$ and the proportion of fuel that can be removed by treatment $\gamma_{c}$. For each cell $c$, let $\alpha_{c}$ and $\beta_{c}$ denote the coefficients of the linear function $f_{c}(\cdot)$. We obtain the following constraint for the fire propagation
\begin{align*}
  T_{c}^{\text{in}} \leq  T_{c'}^{\text{in}} + \alpha_{c'} X_{c'} + \beta_{c'}, &  & \forall c\in G, \: \forall c' \in \Omega_{c} \tag{H.8}. \label{hof:model:c2_final}
\end{align*}

We can reformulate this model using a graph $G=(V, A)$, where each cell $c = (i,j)$ corresponds to a vertex $v \in V$ and the set of arcs $A$ represents the adjacency relation defined by $\Omega$. We replace $T_{c}^{\text{in}}$ by $a_v$ to denote fire arrival time at vertex $v$, $X_{c}$ by $r_{v}$ to denote the resource allocation decision at vertex $v$, and $\overline{X}$ by $k$ to denote the total treatment budget. If we let $s \in V$ denote the ignition vertex and $t \in V$ the target vertex, we can rewrite the model as
\begin{align*}
(M_{hof}) \quad \text{max.} \quad & a_t  \tag{H.9} \label{hof:model:c0_reformulated}                                                                        \\
\text{s.t.} \quad                 & a_s = 0,  \tag{H.10} \label{hof:model:c1_reformulated}                                                                  \\
                                  & a_v \leq  a_u + \alpha_{u} r_{u} + \beta_{u}, &  & \forall uv \in A, \tag{H.11} \label{hof:model:c2_reformulated_final} \\
                                  & 0 \leq r_{v} \leq 1,                          &  & \forall v \in V,  \tag{H.12} \label{hof:model:c5_reformulated}       \\
                                  & \sum_{v \in V} r_{v} \le k.  \tag{H.13} \label{hof:model:c6_reformulated}
\end{align*}

A comparison reveals that the model described above is a simplified version of the MIP formulation presented in Section~\ref{sec:mip_formulation}. One difference is that the suppression variables $r_v$ lack a temporal dimension, such that resources have no release times. Also, fire propagation from a vertex $u$ to a vertex $v$ depends only on the presence of a resource at $u$, rather that the specific direction of the arc $uv$.

Problem HWSP (defined in Section~\ref{def:hwsp}) allows for a set $D\subseteq V$ of target vertices. Model $M_{hof}$ could be extended by adding a new continuous variable $\gamma$, constraints $\gamma \geq a_v$ for $v \in D$, and changing the objective function to $\max \; \gamma$. Problem HWSP also considers integral allocation decisions, which can be obtained by restricting the variables $a_v$ to be binary. Both modifications were explored by \textcite{Hof/2000}.

We now analyze the model proposed by \textcite{Wei/2011}. The notation is the same used in the previous model and Section~\ref{sec:mip_formulation}, with the addition of two parameters: each vertex $v$ has a value $w_v$, which is lost if the vertex burns, and a predicted flame length $F_v$, which is typically estimated using fire simulators. There is also a flame length threshold $\bar F$ that indicates the maximum flame length that can be suppressed, and a maximum number of resources $k$ that can be allocated. \textcite{Wei/2011} do not consider resource release times, so the resource allocation variables $r_v$ do not have a time dimension. The model is
\begin{align*}
    (M_{wei}) \quad \text{min.} \quad & \sum_{v \in V} w_v y_v  \tag{W.1} \label{wei:model:c1}                              \\
    \text{s.t.} \quad                 & a_s = 0,  \tag{W.2} \label{wei:model:c2}                                             \\
                                      & a_v \leq  a_u + t_{uv} + \Delta r_u, &  & uv \in A \tag{W.3}, \label{wei:c3}          \\
                                      & y_v \geq 1 - a_v / H,                &  & v \in V \tag{W.4}, \label{wei:c4}           \\
                                      & \sum_{v \in V} r_{v} \le k,  \tag{W.5} \label{wei:c5}                                \\
                                      & r_v = 0,                             &  & v \in V, F_v > \bar F, \tag{W.6} \label{wei:c6} \\
                                      & r_v, y_v \in \{0,1\},                &  & v \in V.
\end{align*}

The objective function minimizes the total value lost in the vertices that burn. Constraints~\ref{wei:model:c2} and~\ref{wei:c3} are standard fire propagation constraints. Constraint~\ref{wei:c4} sets variable $y_v$ to $1$ if the vertex $v$ burns before the time horizon $H$, and to $0$ otherwise. Constraint~\ref{wei:c5} ensures that the total number of resources allocated to the vertices $v$ is not greater than $k$. \textcite{Wei/2011} are the first to consider firefighter safety constraints. They do this by fixing the value of $r_v$ to $0$ if the predicted flame length $F_v$ at vertex $v$ is greater than the threshold $\bar F$, as shown in constraint~\ref{wei:c6}.

\section{Random Search Algorithm} \label{app:new_instances:random_search}
Algorithm~\ref{alg:random_search} presents a pseudocode for the Random Search algorithm. The idea is to build an allocation $\A$ incrementally from the first decision point to the last. At each decision point $t \in \{t_1, \dots, t_T\}$, we build a set $C$ of candidate vertices that are unprotected and not burned yet. We then add $|R_t|$ randomly selected vertices from $C$ to allocation $\A$. If the new allocation $\A$ is better than the current best allocation $\A^{*}$, we update $\A^{*}$. The algorithm stops when the stopping criterion is met.

\begin{algorithm}[H]
    \SetAlgoLined
    \DontPrintSemicolon
    \caption{\algname{RandomSearch}} \label{alg:random_search}
    $\A^{*} \gets \A_0$\;
    \While{not stop}{
        $\A \gets \A_0$\;
        \For{$t \in \{t_1, \dots, t_{T}\}$}{
          $C \gets V \setminus (B^{\A}_{t} \cup P^{\A})$\;
          \For{$i \in R_t$}{
            $v \gets \text{Pick a random element from } \: C$\;
            $\A_i \gets v$\;
            $C \gets C \setminus \{v\}$\;
          }
        }
        \If{$|B^{\A}_H| < |B^{\A^{*}}_H|$}{
          $\A^{*} \gets \A$\;
        }
    }
    \Return $\A^{*}$\;
  \end{algorithm}

  \printbibliography

\end{document}